\documentclass[11pt, draftcls, onecolumn]{IEEEtran}

\usepackage{amsmath,amssymb,epsfig,ifpdf,cite}
\usepackage{threeparttable}
\usepackage{subfigure,color}

\newcommand{\PfSchFigDir}{.}
\newcommand{\TexComDir}{.}

\newtheorem{cor}{Corollary}
\newtheorem{lem}{Lemma}

\newcommand{\Append}[1]{{Appendix~\ref{#1}}}

\newcommand{\Corollary}[1]{{Corollary~\ref{#1}}}
\newcommand{\Fig}[1]{{Fig.~\ref{#1}}}
\newcommand{\Lemma}[1]{{Lemma~\ref{#1}}}
\newcommand{\Section}[1]{{Section~\ref{#1}}}
\newcommand{\Table}[1]{{Table~\ref{#1}}}

\newcommand{\CN}{\mathcal{CN}}

\newcommand{\Ex}{\mbox{$\mathbb{E}$}}

\newcommand{\Gcal}{\mathcal{G}}

\newcommand{\AbsSr}{{| \hspace{-0.5mm} S_r \hspace{-0.4mm} |}}

\newcommand{\FnRsumNq}{{I_1}}
\newcommand{\FnRsumQ}{{I_2}}

\newcommand{\Fu}{{F_{\hspace{-0.1cm} {}_U}}}

\newcommand{\Fw}{{F_{\hspace{-0.1cm} {}_W}}}

\newcommand{\Fx}{{F_{\hspace{-0.1cm} {}_X}}}

\newcommand{\FxGenCond}{{F_{\hspace{-0.1cm} {}_{X | \textrm{cond}}}}}

\newcommand{\FxCondSr}{{F_{{}_{X | \; {}{k_r^*=k, \AbsSr=n}}}}}

\newcommand{\Fy}{{F_{\hspace{-0.1cm} {}_Y}}}

\newcommand{\Fyk}{{F_{{\hspace{-0.1cm} {}_{Y_k}}}}}
\newcommand{\Fykr}{{F_{{\hspace{-0.1cm} {}_{Y_{k,r}}}}}}
\newcommand{\fz}{{f_{{\hspace{-0.1cm} {}_{Z}}}}}
\newcommand{\Fz}{{F_{{\hspace{-0.1cm} {}_{Z}}}}}

\newcommand{\Fzk}{{F_{{\hspace{-0.1cm} {}_{Z_k}}}}}

\newcommand{\GammaTilde}{{\widetilde{\Gamma}}}

\newcommand{\Hkr}{{H_{k,r}}}

\newcommand{\LFb}{{L_{{}_{\textrm{FB}}}}}

\newcommand{\MinRankingInPf}{{\NFb}}

\newcommand{\NcolInPf}{{\Nrb}}
\newcommand{\NFb}{{N_{\hspace{-0.05cm}{}_{\textrm{FB}}}}}

\newcommand{\Nrb}{{N_{\hspace{-0.05cm}{}_{\textrm{RB}}}}}

\newcommand{\NrowVal}{{n}}

\newcommand{\Nus}{{N_{\hspace{-0.05cm}{}_{\textrm{US}}}}}

\newcommand{\NQb}{{N_{\hspace{-0.05cm}{}_{\textrm{QB}}}}}

\newcommand{\NTx}{{N_{{}_{\textrm{T}}}}}

\newcommand{\PxGenCondQuan}{{P_{\hspace{-0.1cm}{}^{X^{\QuanSup} | \textrm{cond}}}}}

\newcommand{\QuanSup}{{{}^{\textrm{Q}}}}

\newcommand{\RFb}{{R_{{}_{\textrm{FB}}}}}

\newcommand{\Rkr}{{R_{k,r}}}

\newcommand{\Rsum}{{R_{{}_{\textrm{SUM}}}}}
\newcommand{\RsumRatio}{{R_{{}_{\textrm{SUM}}}^{\textrm{ratio}}}}

\newcommand{\Ukr}{{U_{k,r}}}
\newcommand{\UkrQuan}{{U_{k,r}^{{}^{\textrm{Q}}}}}
\newcommand{\Umn}{{U_{m,n}}}
\newcommand{\UmnQuan}{{U_{m,n}^{{}^{\textrm{Q}}}}}

\newcommand{\Wkr}{{W_{k,r}}}

\newcommand{\XrQuan}{{X_{r}^{{}^{\textrm{Q}}}}}
\newcommand{\Ykr}{{Y_{k,r}}}
\newcommand{\Zkr}{{Z_{k,r}}}

\newcommand{\AWGN}{{\textrm{AWGN}}}

\newcommand{\CDD}{{\textrm{CDD}}}
\newcommand{\CDF}{{\textrm{CDF}}}

\newcommand{\CQI}{{\textrm{CQI}}}

\newcommand{\FB}{{\textrm{FB}}}

\newcommand{\MAC}{{\textrm{MAC}}}

\newcommand{\MIMO}{{\textrm{MIMO}}}

\newcommand{\MISO}{{\textrm{MISO}}}

\newcommand{\MRC}{{\textrm{MRC}}}

\newcommand{\OFDMA}{{\textrm{OFDMA}}}

\newcommand{\OSTBC}{{\textrm{OSTBC}}}

\newcommand{\PDF}{{\textrm{PDF}}}

\newcommand{\PMF}{{\textrm{PMF}}}

\newcommand{\RB}{{\textrm{RB}}}

\newcommand{\SISO}{{\textrm{SISO}}}

\newcommand{\SNR}{{\textrm{SNR}}}

\newcommand{\TAS}{{\textrm{TAS}}}

\newcommand{\dB}{{\textrm{dB}}}

\newcommand{\eg}{{\it{e.g.,}}}

\newcommand{\Fstep}{{\textrm{Step}}}

\newcommand{\ie}{{\it{i.e.,}}}

\newcommand{\iid}{{\it{i.i.d.}}}

\newcommand{\st}{{\it{s.t.}}}
\newcommand{\OneColumn}[1]{#1}
\newcommand{\TwoColumn}[1]{}

\newcommand{\bItem}{\begin{itemize}}
\newcommand{\eItem}{\end{itemize}}

\newcommand{\bEq}{\begin{equation}}
\newcommand{\eEq}{\end{equation}}
\newcommand{\bEqn}{\begin{equation*}}
\newcommand{\eEqn}{\end{equation*}}
\newcommand{\bEqnarray}{\begin{eqnarray}}
\newcommand{\eEqnarray}{\end{eqnarray}}

\newcommand{\bCenter}{\begin{center}}
\newcommand{\eCenter}{\end{center}}
\begin{document}

\title{Sum rate analysis of a reduced feedback \OFDMA{} system employing joint scheduling and diversity}
\author
{
    Seong-Ho (Paul) Hur{\hspace{0.08cm}${}^*$}
    \IEEEmembership{Student Member, IEEE,}
    Bhaskar~D.~Rao,
    \IEEEmembership{Fellow, IEEE,}
\thanks
{
    ${}^*$: Corresponding author.
}%
\thanks
{
    S.~H.~(Paul)~Hur and B.~D.~Rao are with the Department of Electrical and Computer Engineering, University of California, San Diego (UCSD), 9500 Gilman Drive La Jolla, CA 92093 (e-mail: shhur@ucsd.edu, brao@ucsd.edu).
}
}
\maketitle

\begin{abstract}
\label{sAbstract}
We consider joint scheduling and diversity to enhance the benefits of multiuser diversity in an \OFDMA{} system.
The \OFDMA{} spectrum is assumed to consist of $\Nrb$ resource blocks and the reduced feedback scheme consists of each user feeding back channel quality information (\CQI) for only the best-$\NFb$ resource blocks.
Assuming largest normalized \CQI{} scheduling and a general value for $\NFb$, we develop a unified framework to analyze the sum rate of the system for both the quantized and non-quantized \CQI{} feedback schemes.
Based on this framework, we provide closed-form expressions for the sum rate for three different multi-antenna transmitter schemes; Transmit antenna selection (\TAS), orthogonal space time block codes (\OSTBC) and cyclic delay diversity (\CDD).
Furthermore, we approximate the sum rate expression and determine the feedback ratio $(\frac{\NFb}{\Nrb})$ required to achieve a sum rate comparable to the sum rate obtained by a full feedback scheme.
\end{abstract}


\OneColumn
{
    \begin{center}
        \bfseries EDICS Category
    \end{center}
    {\small
    \begin{tabular}{l}
    \hspace{1.0cm} MSP: \MIMO{} Communications and signal processing\\
    \hspace{1.0cm} MSP-CAPC: \MIMO{} capacity and performance\\
    \hspace{1.0cm} MSP-MULT: \MIMO multi-user and multi-access schemes
    \end{tabular}
    }

    \newpage
}

\section{Introduction}
\label{sIntroduction4RsumAnal}
Diversity is a common technique employed to mitigate the harmful effects of fading in a wireless channel and to achieve reliable communication \cite{bSklarCm972,bBrennanProc03,bMietznerScom09}.
This is achieved by creating and combining independent multiple copies of a signal between a transmitter and a receiver over various dimensions such as  time, frequency and space \cite{bSklarCm972,bBrennanProc03,bMietznerScom09}.
On the other hand, when fading is viewed in a multiuser communication context and scheduling of users is introduced for sharing the common resources, multiuser diversity can be exploited to significantly increase the system throughput \cite{bKnoppIcc95,bViswanathTit02}.
To exploit multiuser diversity inherent in a wireless network with multiple users, it is necessary to schedule a transmission, at any scheduling instant, to a user with the best channel condition \cite{bKnoppIcc95,bViswanathTit02}, which is also known as opportunistic scheduling \cite{bLiuCn03}.
However, fairness becomes an issue in a system with asymmetric user fading statistics which leads to channel resources being dominated by strong users \cite{bViswanathTit02}.
In order to provide fairness, in addition to exploiting multiuser diversity, a normalized signal to noise ratio (\SNR{})-based or channel quality information (\CQI{})-based scheduling scheme is considered \cite{bHarthiTwc07}.
This can be regarded as a form of proportional fair scheduling \cite{bChoiTvt07}.

The gain from multiuser diversity usually increases with the number of independent users in a system and with a large dynamic range for the channel fluctuation within the time of the scheduling window \cite{bViswanathTit02,bHurPrep10}.
To enhance the sum rate of a system, joint consideration of scheduling and traditional diversity schemes such as transmit antenna selection (\TAS) and maximal ratio combining (\MRC) at a receiver is addressed in \cite{bChenTcom06,bPughTsp10} and the references therein.
The basic principle of joint consideration is to enhance multiuser diversity by increasing the number of independent candidates for selection directly proportional to the number of transmit antennas \cite{bChenTcom06,bPughTsp10}, or by increasing the variation in the channels between the transmitter and receivers as in the opportunistic beamforming methods \cite{bViswanathTit02,bHurPrep10}.
For the purpose of user scheduling and rate adaptation at the transmitter, information about the channel quality has to be fed back to the transmitter by the receivers.
As the number of users as well as the antennas at the transmitter increases, the amount of feedback becomes large placing an enormous burden on the feedback link traffic.
In particular, the amount of feedback may become prohibitive when we consider \OFDMA{} systems which have emerged as the basic physical layer communication technology to meet the high data rate services in future wireless communication standards \cite{bTs36.201.201003}.
With the goal of exploiting frequency diversity in user scheduling, subcarriers in \OFDMA{} systems are grouped into resource blocks and used as the basic unit for user scheduling \cite{bTs36.201.201003}.
When we consider joint scheduling and diversity in \OFDMA{} systems, feedback may be needed for all the resource blocks as well as the antennas, which may easily overwhelm the feedback link traffic even for a system with a small number of users.
This motivates our research into schemes with reduced feedback.

Feedback reduction has received much interest in wireless communications research \cite{bLoveJsac08}.
There are two main methods: feedback rate reduction related to quantization, and feedback number reduction related to reducing the number of parameters being fed back.
See, for example,  \cite{bErikssonPrcd07,bHasselTwc07} and references therein.
For the feedback number reduction, a threshold-based technique is usually considered, so that only the users with a large probability of being scheduled feedback their information \cite{bGesbertIcc04}.
Let $\Nrb$ denote the total number of resource blocks in \OFDMA{} systems or spatial degrees of freedom in a space division multiple access system.
The feedback number reduction can be obtained by letting users feed back information about only the best-$\NFb$ blocks or fewer modes when $\NFb$ is smaller than $\Nrb$ \cite{bJungIscit07,bLeinonenTwc09,bPughTsp10,bChoiTmc10}.
For \OFDMA{} systems employing joint scheduling and diversity, the performance of schemes employing feedback about the best-$\NFb$ blocks, for a general $\NFb$, has not been rigorously studied.
Only the performance for a best-$1$ feedback ($\NFb =1$) or a full feedback scheme ($\NFb = \Nrb$) without consideration of diversity options are given in \cite{bJungIscit07}.
The analysis in \cite{bChoiTmc10} is for a general $\NFb$.
However, it deals with a single-input single-output (\SISO) system with quantized \CQI{} feedback and consequently does not consider the various multi-antenna diversity techniques.

In this paper, we consider an \OFDMA{} system employing joint scheduling as well as using a multi-antenna transmit diversity technique.
Various diversity options are considered in this work; Transmit antenna selection (\TAS), orthogonal space time block codes (\OSTBC) and cyclic delay diversity (\CDD).
For rate adaptation and user scheduling, we assume that users feedback to the transmitter the \CQI{} values of the best-$\NFb$ resource blocks out of a total of $\Nrb$ values.
For a practical variant of the feedback system, we also consider quantized \CQI{}.
The transmitter schedules a transmission in each resource block to a user with largest normalized \CQI{} among users who provided feedback, where normalization is considered to assure fairness across users.
We develop a unified framework consisting of four steps to analyze the sum rate of the system with partial feedback of either non-quantized or quantized \CQI{} for a general $\NFb$, and present closed-form expressions.

Our results show that the performance gap between a full feedback scheme and a best-$1$ ($\NFb = 1$) feedback scheme is not negligible even when there are a moderate number of users.
Then the question arises as to how many \CQI{} values should be fed back to the transmitter to make the gap negligible while minimizing uplink feedback overhead.
This issue is also addressed in our work based on the derived equations for the sum rate.
Specifically, we approximate the sum rate ratio, \ie{} the ratio of the sum rate obtained by a partial feedback scheme to the sum rate obtained by a full feedback scheme.
We express the sum rate ratio as a function of the feedback ratio $(\frac{\NFb}{\Nrb})$, \ie{} the amount of feedback normalized by the total number of blocks.
We show that the sum rate ratio is approximately the same as the probability of the complement of a scheduling outage which corresponds to the case that no user provides \CQI{} to the transmitter for a certain block.
This enables us to provide a simple equation to determine the required feedback ratio for a pre-determined sum rate ratio.
In the case of quantized \CQI{} feedback, we also discuss a feedback design strategy to enhance the sum rate under a fixed feedback load.

In summary, the paper has three main contributions.
First, we present the cumulative distribution function (\CDF) for the \SNR{} of a selected user in the best-$\NFb$ feedback system.
This result has a convenient form in terms of a {\it polynomial} of the \CDF{} of each user's \CQI{}, which is {\it amenable} to further analytical evaluation.
Second, we develop a unified framework to analyze the sum rate of a reduced feedback \OFDMA{} system employing joint scheduling and diversity, and derive closed-form expressions for both the non-quantized and quantized \CQI{} feedback schemes.
Third, we approximate the sum rate result and develop an analytical and simple expression for the required feedback ratio to achieve a pre-determined sum rate ratio.

This paper is organized as follows.
In \Section{sSystemModelAndOverview}, we describe the system model and provide an overview of the unified framework for the analysis.
In \Section{sPerformanceAnalysis4Tas}, we develop the framework and analyze the sum rate of the \TAS{} scheme.
In \Section{sPerformanceAnalysis4OstbcCdd}, we analyze the sum rate for both \OSTBC{} and \CDD{} schemes employing the framework.
In \Section{sRelation}, we develop the relation between the sum rate ratio and feedback ratio, and derive the expression for the required feedback ratio.
In \Section{sNumericalResults4RsumAnal}, we show numerical results and they support the analytical results.
We conclude in \Section{sConclusion4RsumAnal}.

\section{System Model and Overview of the framework}
\label{sSystemModelAndOverview}
In this section, we first describe the system model and then provide an overview of the unified framework for the sum rate analysis.

\subsection{System model}
\label{sSystemModel4RsumAnal}
We consider a multiple-input single-output (\MISO) complex Gaussian broadcast channel with one base station equipped with $\NTx$ transmit antennas and $\Nus$ users each equipped with a single antenna, as shown in \Fig{fSysBldForRsumAnal}.
An \OFDMA{} system is assumed.
In a multiuser \OFDMA{} system the throughput is larger when the resource allocation is flexible and has high granularity, \eg{} assignment at the individual subcarrier level.
However, the complexity and feedback overhead can be prohibitive, calling for simpler approaches.
In our work, the overall subcarriers are grouped into $\Nrb$ resource blocks ($\RB$), and each block contains contiguous subcarriers.
The assignment is done at the block level, \ie{} a resource block is assigned to a user.
The block size is assumed to be known and in practice can be determined at the medium access control (\MAC) layer taking into account the number of users.
For this system, we showed in \cite{bHurPrep10} that the optimal channel selectivity maximizing the sum rate is flat within each block and independent across blocks.
We assume the optimal channel selectivity condition in our analysis of the system performance.\OneColumn
{
   \begin{figure}
       \centering
           \includegraphics[width=12.0cm]{\PfSchFigDir/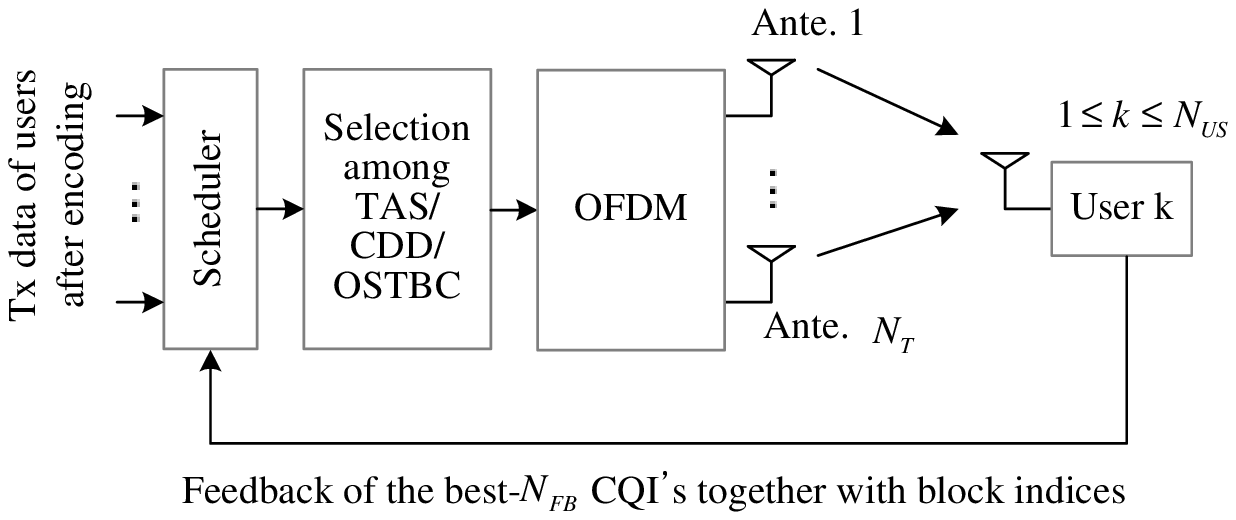}
           \caption{System block diagram of a multiuser \OFDMA{} system.}
           \label{fSysBldForRsumAnal}
   \end{figure}
}\TwoColumn
{
   \begin{figure}
       \centering
           \includegraphics[width=8.0cm]{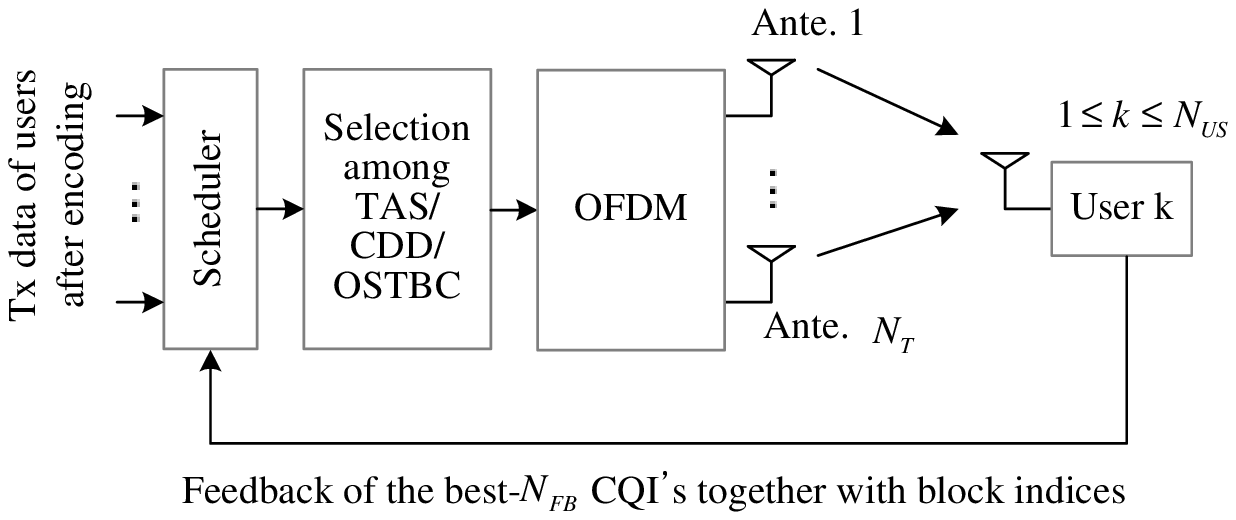}
           \caption{System block diagram of a multiuser \OFDMA{} system.}
           \label{fSysBldForRsumAnal}
   \end{figure}
}

Let $H_{k,r,i}$ denote the channel between transmit antenna-$i$ and the receive antenna of user-$k$ for resource block-$r$, where $1 \leq k \leq \Nus$, $1 \leq r \leq \Nrb$ and $1 \leq i \leq \NTx$.
We assume that $H_{k,r,i}$ follows a complex Gaussian distribution, \ie{} $\CN(0,c_k)$,\footnote{$\CN (\mu,$$\sigma^2)$ denotes a circularly symmetric complex Gaussian distribution with mean $\mu$ and variance $\sigma^2$.} where $c_k$ denotes the average channel power of user-$k$ and reflects the fact that the users are distributed asymmetrically.
We assume that $c_k$ for each user is known to the transmitter by infrequent feedback from users.
We also assume that $H_{k,r,i}$ is independent across users ($k$), blocks ($r$) and transmit antennas ($i$).
Then, the received signal of user-$k$ at block-$r$ satisfies the equation
    \begin{equation}
        \label{eChannelForRsumAnal}
        y_{k,r} = H_{k,r} \; s_{k,r} + n_{k,r}
    \end{equation}
where $s_{k,r}$ is the transmitted symbol and $n_{k,r}$ is additive white Gaussian noise (\AWGN) with $\CN (0,\sigma_w^2)$. We note that $\Hkr$ is the equivalent channel depending on the specific diversity technique and is a function of $H_{k,r,i}$, which will be shown in later sections.

For reliable and adaptive communication, the knowledge of the channel between the transmitter and receiver is required at the transmitter.
For this purpose, we assume that channel quality information (\CQI{}) of resource blocks is fed back from users to the transmitter.
The feedback policy is that users measure \CQI{} for each block at their receiver and feed back the \CQI{} values of the best-$\NFb$ resource blocks from among the total $\Nrb$ values \cite{bJungIscit07}.
Since we assume that the users are asymmetrically distributed in their average \SNR{}, scheduling is based on \CQI{} normalized by each user's mean value at the transmitter.
For each block, the user with largest normalized \CQI{} is chosen from among the users who fed back \CQI{} to the transmitter for that block.
If no user provides \CQI{} for a certain block, \ie{} the case of a scheduling outage in the block \cite{bJungIscit07}, we assume that the transmitter does not utilize that block.
However, one can easily incorporate other variations such as round-robin scheduling or a scheduling scheme which maintains the previously assigned user.
For diversity, we consider three different multiple transmit antenna techniques; transmit antenna selection (\TAS) \cite{bChenTvt05}, cyclic delay diversity (\CDD) \cite{bDammannGcom01},\footnote{For \CDD{}, we consider that phases are multiplied on the basis of a block to maintain the characteristic of flat fading inside a block. In a strict sense, the scheme we consider is classified as \CDD{} when the block consists of a single subcarrier, and as the frequency domain opportunistic beamforming when the block consists of more than one subcarriers \cite{bViswanathTit02}.} and orthogonal space time block codes (\OSTBC) \cite{bAlamouti98}.
Let $\Zkr$ denote \CQI{} of user-$k$ at block-$r$, which will be the starting point of the analysis.
Then, $\Zkr$ depends on the diversity technique, the noise variance and channel $H_{k,r,i}$.

As the number of users increases, the amount of feedback will be prohibitive for a full feedback scheme, \ie{} \CQI{} feedback for all the resource blocks, so that we focus on the sum rate for partial feedback schemes with a general $\NFb$.
Instead of investigating the asymptotic property of the sum rate for a very large or infinite number of users \cite{bSharifTit05,bSongTcom06}, we focus on the exact sum rate for the system with a finite number of users.
Specifically, we develop a unified framework consisting of four steps to analyze the sum rate of this system with partial feedback of either non-quantized or quantized \CQI{}, and present closed-form expressions.
An overview of the framework is provided next.

\subsection{Overview of the unified framework}
\label{sGeneralApproach}
The framework for the sum rate analysis consists of four steps, where the $n$-th step is denoted as \Fstep-$n$.
We first discuss the analysis in the non-quantized \CQI{} case.
We find $\Fzk$ in \Fstep-$1$, \ie{} the \CDF{} of $\Zkr$ which is the \CQI{} of user-$k$ at block-$r$ at a receiver.\footnote{Since we assume that blocks are identically distributed, for notational simplicity, we omit $r$ in ${{F_{{\hspace{-0.1cm} {}_{Z_{k,r}}}}}}$, which is also the case for other notations of \CDF{}s.}
This depends on the choice of the diversity technique.
We find $\Fyk$ in \Fstep-$2$, \ie{} the \CDF{} of $\Ykr$ denoting the \SNR{} of user-$k$ for resource block-$r$ as seen by the transmitter as a consequence of partial feedback.
We find $\FxGenCond$ in \Fstep-$3$, \ie{} the conditional \CDF{} of $X_{r}$ denoting the \SNR{} of a selected user as a consequence of scheduling.
The conditioning in \Fstep-$3$ is related to the asymmetric user distribution in their average \SNR{} and the number of contending users for the block.
The important characteristics of $\FxGenCond$ is that it has a convenient form in terms of a {\it polynomial in $\Fzk$}, which is {\it amenable} to further integration to obtain the sum rate in \Fstep-$4$.
Thus, once we find $\Fzk$ and we have the integration result for a throughput equation with respect to an arbitrary power of $\Fzk(x)$, \ie{} $\int_0^\infty \log_2(1+x) \; d\{ \Fzk(x) \}^n$ for an arbitrary positive integer $n$, we can obtain closed-form sum rate expressions in a straightforward manner.

In the quantized \CQI{} case, following the same approach as the first two steps in the non-quantized case, we find $\Fw$, the \CDF{} of $\Wkr$ denoting the normalized \CQI{} at a receiver and $\Fu$, the \CDF{} of $\Ukr$ denoting the normalized \CQI{} as seen by the transmitter.
Then, we find $\PxGenCondQuan$ in \Fstep-$3$, \ie{} the conditional probability mass function (\PMF{}) of $\XrQuan$, the \SNR{} of a selected user.
By taking an average of throughputs over the \PMF{} found, we can obtain closed-form sum rate expressions in \Fstep-$4$.
For easy reference, we summarize the steps in \Table{tMainStepsForRsumAnal}.
\OneColumn
{
    \begin{table}[h]
        \centering
        \begin{threeparttable}
            \centering
            \caption{The main steps for the unified framework to obtain the sum rate.}
            \label{tMainStepsForRsumAnal}
            \begin{tabular}[c]{|c||c|c||c|c|}
                \hline
                                \begin{tabular}[c]{c} {\it Framework} \end{tabular}   &  \multicolumn{2}{|c||}{Non-quantized \CQI{} feedback} & \multicolumn{2}{|c|}{Quantized \CQI{} feedback} \\
                \cline{2-5}
                                &   Random variable   &     Output{}  &
                                    Random variable   &     Output{} \\
                \hline
                        \Fstep-$1$ &   \hspace{-1.1cm} $\Zkr$: \CQI{} at a receiver   &      $\Fzk$  &
                                   $\Wkr$: Normalized \CQI{} at a receiver   &     $\Fw$ \\
                \hline
                        \begin{tabular}[c]{c} \Fstep-$2$    \end{tabular} &   $\Ykr$: \SNR{} seen at a transmitter   &   $\Fyk$  &
                                   \hspace{-0.0cm} $\Ukr$: Normalized \CQI{} seen at a transmitter   &   $\Fu$ \\
                \hline
                        \begin{tabular}[c]{c} \Fstep-$3$    \end{tabular} &   \hspace{-0.3cm} $X_{r}$: \SNR{} of a selected user    &      $\FxGenCond$  &
                                   \hspace{-0.7cm} $\XrQuan$: \SNR{} of a selected user  &      $\PxGenCondQuan$ \\
                \hline
                        \Fstep-$4$ &   \multicolumn{2}{|c||}{\hspace{-0.2cm} $ \Ex_{\textrm{cond}} \Ex_{{}_{X_r}}[\log_2 (1+X_r) | \textrm{cond}] $ } & \multicolumn{2}{|c|}{\hspace{-0.2cm} $\Ex_{\textrm{cond}} \Ex_{\XrQuan}[\log_2 (1+\XrQuan) | \textrm{cond}]$ }   \\
                \hline
            \end{tabular}
            \begin{tablenotes}
                \item[] $k$: user index, $r$: block index.
            \end{tablenotes}
        \end{threeparttable}
    \end{table}
}
In summary, \Fstep-$1$ of the unified framework depends on the diversity technique.
The next two steps (\Fstep-$2$ and \Fstep-$3$) depend on the feedback and scheduling policy.
\Fstep-$4$ involves evaluating the performance measure.
We explain the procedure by providing details of the four steps for the \TAS{} scheme in \Section{sPerformanceAnalysis4Tas}.
Then in \Section{sPerformanceAnalysis4OstbcCdd}, we focus on finding the \CDF{} of $\Zkr$ in \Fstep-$1$ for \OSTBC{} and \CDD{}.
\Fstep-$2$ and \Fstep-$3$ do not require much additional effort, and we provide the sum rate result utilizing \Fstep-$4$.

\section{Sum rate analysis with application to \TAS{}}
\label{sPerformanceAnalysis4Tas}
In this section, we explain the details of the framework, consisting of the four steps in \Table{tMainStepsForRsumAnal}, with application to the transmit antenna selection (\TAS)-based diversity scheme for both non-quantized \CQI{} and quantized \CQI{}.

\subsection{Sum rate analysis for non-quantized \CQI}
\label{sNonQSumRate4Tas}
\subsubsection{\Fstep-$1$, finding $\Fzk(x)$}
\label{sStepOne4Nq}
This step consists of finding the distribution of \CQI.
In \TAS{}, a transmit antenna with the best channel condition among all the transmit antennas is selected for transmission \cite{bChenTvt05}.
Thus, the equivalent channel at block-$r$ of user-$k$ is a channel with maximum \CQI{} across transmit antennas, \ie{} $H_{k,r}= H_{k,r,i^*}$ where $i^* = \arg \max_{1 \leq i \leq \NTx} |H_{k,r,i}|^2$.
Since we assume that $H_{k,r,i}$ follows $\CN(0,c_k)$, $|H_{k,r,i}|^2$ follows the Gamma distribution $\Gcal(1,\frac{1}{c_k})$ \cite[(17.6)]{bJohnsonWiley94}.
Here, $\Gcal(\alpha, \beta)$ denotes the Gamma distribution whose \CDF{} is given by \cite[(17.3)]{bJohnsonWiley94}
    \begin{equation}
        \label{eGammaCdf}
        F(x) = \GammaTilde (\alpha, \beta x) = \frac{1}{\Gamma(\alpha)} \int_0^{\beta x} t^{\alpha-1} e^{-t} dt,
    \end{equation}
where $\GammaTilde(\cdot,\cdot)$ is the incomplete Gamma function ratio given by $\GammaTilde(a,x) = \frac{1}{\Gamma(a)} \int_0^{x} t^{a-1} e^{-t} dt$ \cite[(17.3)]{bJohnsonWiley94} and $\Gamma(\cdot)$ is the Gamma function given by $\Gamma(a) = \int_0^{\infty} t^{a-1} e^{-t} dt$ \cite{bGradshteynAp00}.
Then, equivalent \CQI{} in \TAS{} is $\Zkr = |H_{k,r}|^2 = \max_{1 \leq i \leq \NTx} |H_{k,r,i}|^2$.
From the assumption of the independent and identical distribution (\iid{}) for $H_{k,r,i}$'s in $i$, the \CDF{} of $\Zkr$ is given by\OneColumn
{
    \begin{equation}
        \label{eFzkTas}
        \hspace{-0.0cm} \Fzk(x) = \Pr \big \{ \Zkr \leq x \big \} \stackrel{(a)}{=}  \big[ \Pr \big \{ |H_{k,r,i}|^2 \leq x \big \} \big]^\NTx \stackrel{(b)}{=}  \big[\GammaTilde(1, \tfrac{x}{c_k})\big]^\NTx
    \end{equation}
}\TwoColumn
{
    \begin{equation*}
        \hspace{-1.0cm} \Fzk(x) = \Pr \big \{ \Zkr \leq x \big \} \stackrel{(a)}{=}  \big[ \Pr \big \{ |H_{k,r,i}|^2 \leq x \big \} \big]^\NTx
    \end{equation*}
    \begin{equation}
        \label{eFzkTas}
        \hspace{1.0cm} \stackrel{(b)}{=}  \big[\GammaTilde(1, \tfrac{x}{c_k})\big]^\NTx
    \end{equation}
}
where $(a)$ follows from the order statistics \cite[2.1.1]{bDavidJwas04} that $\Zkr$ is the maximum of independent $|H_{k,r,i}|^2$s, and $(b)$ follows from the fact that $|H_{k,r,i}|^2$ has the distribution $\Gcal(1,\frac{1}{c_k})$.
We note that the \SNR{} at block-$r$ of user-$k$ is $ \SNR_{k,r}=\rho \Zkr$ where $\rho = P / \sigma_w^2$ when the total transmit power is $P$.

\subsubsection{\Fstep-$2$, finding $\Fyk(x)$}
\label{sStepTwo4Nq}
This step considers the distribution of \CQI{} as a result of partial feedback.
As a reminder, each user feeds back the best-$\NFb$ \CQI{} values to the transmitter.
Let $Z_{k,(\ell)}$ denote the order statistics of $\Zkr$'s of user-$k$, where $Z_{k,(1)} \leq \cdots \leq Z_{k,(\NcolInPf)}$.
Then, the feedback scheme is equivalent to each user determining the order statistics for its \CQI{} and feeding back \CQI{} $Z_{k,(\ell)}$'s, for $\Nrb - \MinRankingInPf + 1 \leq \ell \leq \Nrb$ and the corresponding resource block indices.
Let $\Ykr$ denote the \SNR{} corresponding to received \CQI{} at the transmitter for user-$k$ at block-$r$ through feedback.
If user-$k$ provides feedback containing \CQI{} for block-$r$, then based on the \iid{} assumption of $\Zkr$'s in $r$, the \SNR{} $\Ykr$ viewed from the transmitter can be interpreted as any one of the best-$\MinRankingInPf$ values multiplied by $\rho$.
To capture this aspect, let $\Rkr$ denote a random variable with a probability mass function of $\Pr \{ \Rkr = \ell \} = \frac{1}{\MinRankingInPf}$, for $\NcolInPf - \MinRankingInPf + 1 \leq \ell \leq \NcolInPf$.
Then $\Ykr$ is given by
\begin{equation}
    \label{eYk}
    \Ykr = \rho Z_{k,(\Rkr)}.
\end{equation}
The \CDF{} of $\Ykr$, $\Fyk(x)$, is given in the following lemma.
\begin{lem}
    \label{lFyk}
    For $\Fzk(x)$ in \eqref{eFzkTas}, the \CDF{} of $\Ykr$ in \eqref{eYk} is given by
        \begin{equation}
            \label{eFy00}
                \Fyk(x)= \sum_{m=0}^{\MinRankingInPf-1} e_1(\NcolInPf,\MinRankingInPf,m) \{ \Fzk(\tfrac{x}{\rho}) \}^{\NcolInPf-m}
        \end{equation}
    where
        \begin{equation}
            \label{eE1abm}
                e_1(\NcolInPf,\MinRankingInPf,m)= \sum_{\ell=m}^{\MinRankingInPf-1} \tfrac{\MinRankingInPf - \ell}{\MinRankingInPf} \tbinom{\NcolInPf}{\ell} \tbinom{\ell}{m} (-1)^{\ell-m}.
        \end{equation}
\end{lem}
\begin{proof}
    See \Append{sPf4LemFyk}.
\end{proof}
\begin{cor}
    \label{cE1}
    When $\MinRankingInPf= \NcolInPf$ (\ie{} full feedback), $e_1(\NcolInPf,\NcolInPf,m)= 1$ for $m=\NcolInPf-1$, and $0$ otherwise.
\end{cor}
\begin{proof}
    See \Append{sPf4CorrE1}.
\end{proof}

For example in best-$1$ feedback ($\NFb=1$), since $e_1(\Nrb,1,m)$ is only non-zero for $m=0$ and the value is $1$, we can verify that \eqref{eFy00} reduces to $\Fyk(x) = \{\Fzk(\tfrac{x}{\rho})\}^\Nrb$, which confirms $\Ykr= \rho \times \max_{1 \leq r' \leq \Nrb} Z_{k,r'}$ \cite[2.1.1]{bDavidJwas04}.
In full feedback ($\NFb = \Nrb$), since $e_1(\Nrb,\Nrb,m)= 1$ for $m= \Nrb-1$ and zero otherwise from \Corollary{cE1}, we can verify that \eqref{eFy00} reduces to $\Fyk(x) = \Fzk(\tfrac{x}{\rho})$, which confirms that $\Ykr= \rho \Zkr$.
That is, $\Ykr$ has the same statistics as $\SNR_{k,r}$ for full feedback.

\subsubsection{\Fstep-$3$, finding the conditional \CDF{} of $X_r$}
\label{sStepThree4Nq}
This step involves finding the distribution of the \SNR{} of the channel of the user selected in the scheduling step based on partial feedback.
Since a channel is assumed to be \iid{} across the resource blocks for each user, the probability that a user provides the transmitter with \CQI{} for block-$r$ is $\frac{\NFb}{\Nrb}$.
Let $S_r$ denote a set of users who provided \CQI{} to the transmitter for block-$r$.
Since the channel is independent across users, the number of users who provided \CQI{} at block-$r$, \ie{} $|S_r|$, follows the binomial distribution with the probability mass function \cite{bGarciaAw94}\OneColumn
{
    \begin{equation}
        \label{ePmfS}
        \Pr \{ |S_r| = n \} = \binom{\Nus}{n} \left( \frac{\NFb}{\Nrb} \right)^n \left( 1 - \frac{\NFb}{\Nrb} \right)^{\Nus-n}, \: 0 \leq n \leq \Nus.
    \end{equation}
}\TwoColumn
{
    for $0 \leq n \leq \Nus$ as
    \begin{equation}
        \label{ePmfS}
        p_n \triangleq \Pr \{ |S_r| = n \} = \binom{\Nus}{n} \left( \frac{\NFb}{\Nrb} \right)^n \left( 1 - \frac{\NFb}{\Nrb} \right)^{\Nus-n}.
    \end{equation}
}

For \Fstep-$3$ related to the {\it user selection} in \Table{tMainStepsForRsumAnal}, let $\Ukr = \frac{\Ykr}{\rho c_k}$, \ie{} normalized \CQI{} of user-$k$ in block-$r$ viewed at the transmitter.
Based on the scheduling policy, a user with the largest $\Ukr$ among users in $S_r$ is scheduled on block-$r$ by the transmitter.
In our assumption, $\Ykr$'s are independent but not identically distributed in $k$ due to the different average \SNR{} distribution (\ie{} different $c_k$) across users.
However, $\Ukr$'s are \iid{} in $k$ as well because they are normalized by their average \SNR{}, \ie{} $\rho c_k$.
Let $k_r^*$ denote a random variable representing a selected user for transmission on block-$r$ by the transmitter and $X_r$ be the \SNR{} of the selected user.
Since, in our model we do not utilize a block when $|S_r|=0$, we concentrate on the case $|S_r| \neq 0$.
Note that $|S_r|=0$ corresponds to a scheduling outage.
Then, it is shown in \Append{sDeriCdfX} that the conditional \CDF{} of $X_r$ is given by
    \begin{equation}
        \label{eFx00}
        \FxCondSr (x)= \left \{ \Fyk (x) \right \}^{n}.
    \end{equation}
Since $\Fyk(x) = \Fzk(\tfrac{x}{\rho})$ for full feedback ($\NFb= \Nrb$) and $\Fyk(x) = \{ \Fzk(\tfrac{x}{\rho}) \}^\Nrb$ for best-$1$ feedback ($\NFb= 1$), for these two special cases we have
    \begin{equation}
        \label{eTwoSpecialCdfs}
        \hspace{-0.4cm} \FxCondSr (x) =   \begin{cases}
                        \left [ \Fzk (\tfrac{x}{\rho}) \right ]^{n} \hspace{0.7cm} \textrm{: Full \FB}\\
                        \left [ \Fzk (\tfrac{x}{\rho}) \right]^{n \Nrb} \hspace{0.2cm} \textrm{: Best-$1$ \FB}
                    \end{cases}
    \end{equation}
with $\Fzk(x)$ given in \eqref{eFzkTas}.
For the general case, substituting $\Fyk(x)$ from \Lemma{lFyk} into \eqref{eFx00}, we have the following result.
\begin{lem}
    \label{lFxk}
    For $\Fzk(x)$ in \eqref{eFzkTas}, the conditional \CDF{} of $X_r$ in \eqref{eFx00} is given by
        \begin{equation}
            \label{eCdfX}
            \FxCondSr (x)= \sum_{m=0}^{n(\NFb-1)} e_2(\Nrb,\NFb,n,m) \big \{ \Fzk(\tfrac{x}{\rho}) \big \}^{n \Nrb - m}
        \end{equation}
    where $e_2(\NcolInPf,\MinRankingInPf,\NrowVal,m)$ is given by\OneColumn
{
        \begin{equation}
            \label{eE2}
            \hspace{-0.0cm} e_2(\NcolInPf,\MinRankingInPf,\NrowVal,m)= \begin{cases}
                \{ e_1(\NcolInPf,\MinRankingInPf,0) \}^\NrowVal, \quad m=0\\
                \frac{1}{m e_1(\NcolInPf,\MinRankingInPf,0)} \sum_{\ell=1}^{\min \{ m, \MinRankingInPf-1 \}} \{(\NrowVal+1) \ell - m\} \\ \hspace{0.3in} \times e_1(\NcolInPf,\MinRankingInPf,\ell) e_2(\NcolInPf,\MinRankingInPf,\NrowVal,m-\ell), \quad 1 \leq m < \NrowVal(\MinRankingInPf-1) \\
                \{ e_1(\NcolInPf,\MinRankingInPf,\MinRankingInPf-1) \}^\NrowVal, \quad m=\NrowVal(\MinRankingInPf-1).
            \end{cases}
        \end{equation}
}\TwoColumn
{
        \begin{equation}
            \label{eE2}
            \hspace{-0.0cm} \begin{cases}
                \{ e_1(\NcolInPf,\MinRankingInPf,0) \}^\NrowVal, \quad m=0\\
                \sum_{\ell=1}^{\min \{ m, \MinRankingInPf-1 \}} \frac{\{(\NrowVal+1) \ell - m\}}{m e_1(\NcolInPf,\MinRankingInPf,0)} \; e_1(\NcolInPf,\MinRankingInPf,\ell) \\ \hspace{0.3in} \times  e_2(\NcolInPf,\MinRankingInPf,\NrowVal,m-\ell), \quad 1 \leq m < \NrowVal(\MinRankingInPf-1) \\
                \{ e_1(\NcolInPf,\MinRankingInPf,\MinRankingInPf-1) \}^\NrowVal, \quad m=\NrowVal(\MinRankingInPf-1).
            \end{cases}
        \end{equation}
}
\end{lem}
\begin{proof}
    See \Append{sPf4LemFxk}.
\end{proof}

\subsubsection{\Fstep-$4$, finding the sum rate}
\label{sStepFour4Nq}
Now we use the derived \CDF{} to obtain the sum rate of the \OFDMA{} system.
Since blocks are identically distributed, the sum rate is $\Rsum = \frac{1}{\Nrb} \sum_{r=1}^{\Nrb} \Ex [ \log (1 + X_r ) ] = \Ex[\log(1+X_r)]$.
From the property of the conditional expectation \cite{bGarciaAw94}, we have\OneColumn
{
    \begin{equation}
        \label{eRsumDef}
        \hspace{-0.0cm} \Rsum = \Ex_{k_r^*} \; \Ex_{{}_{|S_r|}}  \big [ \; \Ex_{{}_{X_r}} \big [ \log(1+X_r) \; | \; |S_r|=0 \big] + \Ex_{{}_{X_r}} \big[ \log(1+X_r) \; | \; |S_r| \neq 0 \big ] \; \big ].
    \end{equation}
}\TwoColumn
{
    \begin{equation*}
        \hspace{-0.0cm} \Rsum = \Ex_{k_r^*} \; \Ex_{{}_{|S_r|}}  \big [ \; \Ex_{{}_{X_r}} \big [ \log(1+X_r) \; | \; |S_r|=0 \big]
    \end{equation*}
    \begin{equation}
        \label{eRsumDef}
        \hspace{-0.0cm} + \Ex_{{}_{X_r}} \big[ \log(1+X_r) \; | \; |S_r| \neq 0 \big ] \; \big ].
    \end{equation}
}
Since $X_r=0$ when $|S_r|=0$, the first term is zero and does not contribute to the sum rate.
Other variations on the scheduling when there is a scheduling outage, as mentioned in \Section{sSystemModel4RsumAnal}, can be readily incorporated into the first term.\OneColumn
{
Concentrating on the second term, the sum rate is further developed as follows:
    \begin{equation*}
         \hspace{-4.5cm} \Rsum = \Ex_{k_r^*} \; \Ex_{{}_{|S_r|}} \left[ \int_0^\infty \log(1+x) \; d \hspace{-0.1cm} \left \{ \FxCondSr (x) \right \} \; \big | \; |S_r|=n \neq 0 \right]
    \end{equation*}
    \begin{equation*}
         \hspace{-2.0cm} \stackrel{(a)}{=} \Ex_{k_r^*} \; \Ex_{{}_{|S_r|}} \bigg[ \sum_{m=0}^{n(\NFb-1)} e_2(\Nrb,\NFb,n,m) \int_0^\infty \log(1+x) \; d \hspace{-0.1cm} \left\{ \Fzk(\tfrac{x}{\rho})  \right\}^{n \Nrb - m}  \; \big | \; |S_r|=n \neq 0 \bigg]
    \end{equation*}
    \begin{equation}
        \label{eSumRateGenFbTas}
        \hspace{-0.1cm} \stackrel{(b)}{=} \tfrac{1}{\Nus} \sum_{k=1}^{\Nus} \sum_{n=1}^{\Nus} \tbinom{\Nus}{n} \left( \tfrac{\NFb}{\Nrb} \right)^n \left( 1 - \tfrac{\NFb}{\Nrb} \right)^{\Nus-n} \; \sum_{m=0}^{n(\NFb-1)} e_2(\Nrb,\NFb,n,m) \FnRsumNq (1,\tfrac{1}{\rho c_k},{(n \Nrb - m)\NTx}),
    \end{equation}
}\TwoColumn
{
Concentrating on the second term, the sum rate is further developed as \eqref{eSumRateGenFbTas} in the next page:
    \begin{figure*}[!t]
    \normalsize
    \begin{equation*}
         \hspace{-5.5cm} \Rsum = \Ex_{k_r^*} \; \Ex_{{}_{|S_r|}} \left[ \int_0^\infty \log(1+x) \; d \hspace{-0.1cm} \left \{ \FxCondSr (x) \right \} \; \big | \; |S_r|=n \neq 0 \right]
    \end{equation*}
    \begin{equation*}
         \hspace{-1.0cm} \stackrel{(a)}{=} \Ex_{k_r^*} \; \Ex_{{}_{|S_r|}} \bigg[ \sum_{m=0}^{n(\NFb-1)} e_2(\Nrb,\NFb,n,m) \int_0^\infty \log(1+x) \; d \hspace{-0.1cm} \left\{ \Fzk(\tfrac{x}{\rho})  \right\}^{n \Nrb - m}  \; \big | \; |S_r|=n \neq 0 \bigg]
    \end{equation*}
    \begin{equation}
        \label{eSumRateGenFbTas}
        \hspace{-0.1cm} \stackrel{(b)}{=} \tfrac{1}{\Nus} \sum_{k=1}^{\Nus} \sum_{n=1}^{\Nus} \tbinom{\Nus}{n} \left( \tfrac{\NFb}{\Nrb} \right)^n \left( 1 - \tfrac{\NFb}{\Nrb} \right)^{\Nus-n} \; \sum_{m=0}^{n(\NFb-1)} e_2(\Nrb,\NFb,n,m) \FnRsumNq (1,\tfrac{1}{\rho c_k},{(n \Nrb - m)\NTx}),
    \end{equation}
    \end{figure*}
}
where $(a)$ follows from the conditional \CDF{} of $X_r$ in \eqref{eCdfX}; $(b)$ follows from the fact that the \PMF{} $\Pr\{ k_r^* = k \} = \frac{1}{\Nus}$, because $\Ukr$ for user selection is \iid{} in $k$, and $\Pr\{ |S_r| = n \} $ is given by \eqref{ePmfS}, and that we have the following integration identity for the \CDF{} $\Fz(x)$ with the form given in \eqref{eGammaCdf}  \cite{bChenTcom06}
    \begin{equation}
        \label{eSumRateInt}
        \int_0^\infty \log (1+x) \; d \hspace{-0.1cm} \left \{ \Fz (x) \right \}^{n}  = \FnRsumNq (\alpha,\beta,n).
    \end{equation}
It is shown in \Append{sDeriC1} that $\FnRsumNq (x,y,z)$ is given by\OneColumn
{
    \begin{equation}
        \label{eC1Ori}
        \tfrac{z}{(x-1)! \ln 2} \sum_{k=0}^{z-1} (-1)^{k} \tbinom{z-1}{k} \sum_{i=0}^{k(x-1)} b_{k,i} \; \tfrac{(x+i-1)!}{(k+1)^{x+i}} \sum_{\ell=0}^{x+i-1} \{(k+1)y\}^{\ell} \; \Gamma(-\ell,(k+1)y) e^{(k+1)y}
    \end{equation}
}\TwoColumn
{
    \begin{equation*}
        \hspace{-2.0cm} \tfrac{z}{(x-1)! \ln 2} \sum_{k=0}^{z-1} (-1)^{k} \tbinom{z-1}{k} \sum_{i=0}^{k(x-1)} b_{k,i} \; \tfrac{(x+i-1)!}{(k+1)^{x+i}}
    \end{equation*}
    \begin{equation}
        \label{eC1Ori}
        \times \sum_{\ell=0}^{x+i-1} \{(k+1)y\}^{\ell} \; \Gamma(-\ell,(k+1)y) e^{(k+1)y}
    \end{equation}
}
where $\Gamma(a,x) = \int_x^{\infty} t^{a-1} e^{-t} dt$ is the incomplete Gamma function \cite[8.350.2]{bGradshteynAp00} and\OneColumn
{
    \begin{equation}
        \label{ebki}
        b_{k,i} = \begin{cases}
            1, \hspace{0.2cm} i=0\\
            \frac{1}{i} \sum_{n=1}^{\min \{ i, x-1 \}} \frac{n(k+1) - i}{n!} \; b_{k,i-n}, \hspace{0.2cm} 1 \leq i < k(x-1) \\
            \frac{1}{[(x-1)!]^k}, \hspace{0.2cm} i=k(x-1)
        \end{cases}.
    \end{equation}
}\TwoColumn
{
    \begin{equation}
        \label{ebki}
        b_{k,i} = \begin{cases}
            1, \hspace{0.2cm} i=0\\
            \frac{1}{i} \sum_{n=1}^{\min \{ i, x-1 \}} \frac{n(k+1) - i}{n!} \; b_{k,i-n}, \hspace{0.2cm} 1 \leq i < k(x-1) \\
            \frac{1}{[(x-1)!]^k}, \hspace{0.2cm} i=k(x-1).
        \end{cases}
    \end{equation}
}
When $x=1$, $\FnRsumNq (x,y,z)$ is further reduced to \cite{bChenTcom06,bJorswieckEurasip09}
    \begin{equation}
        \label{eC1Red}
        \FnRsumNq (1,y,z)= \tfrac{1}{\ln 2} \sum_{k=1}^{z} (-1)^{k-1} \tbinom{z}{k} \Gamma(0,ky) e^{ky}.
    \end{equation}
We note that the conditional \CDF{} of $X_r$ in \eqref{eCdfX} is amenable to the integration since it is represented in terms of a polynomial in $\Fzk(x)$ and we have the integration result in \eqref{eSumRateInt}.
Although we can represent the incomplete Gamma function in \eqref{eC1Ori} using a finite summation as in \cite{bChenTcom06} and \cite{bAbramowitzUs70}, \ie{} $\Gamma(-\ell,(k+1)y) = \frac{(-1)^\ell}{\ell!} \big[ \Gamma(0,(k+1)y) - e^{-(k+1)y} \sum_{m=0}^{\ell-1} \frac{(-1)^m m !}{\{(k+1)y\}^{m+1}} \big]$, we note that the form in \eqref{eC1Ori} is much more appropriate for easy, fast and precise evaluation especially for large $z$, which is related to $\Nrb$, $\Nus$, and $\NTx$.

The expression can be simplified to obtain the sum rate for the special cases of best-$1$ and full feedback.\OneColumn
{
    \begin{equation}
        \label{eTwoSpecialSumRateTas}
        \hspace{-0.0cm} \Rsum =   \begin{cases}
                        \frac{1}{\Nus} \sum_{k=1}^{\Nus} \FnRsumNq (1,\frac{1}{\rho c_k},{ \Nus \NTx}) \hspace{6.3cm} \textrm{: Full \FB}\\
                        \frac{1}{\Nus} \sum_{k=1}^{\Nus} \sum_{n=1}^{\Nus} \binom{\Nus}{n} \left( \frac{1}{\Nrb} \right)^n \left( 1 - \frac{1}{\Nrb} \right)^{\Nus-n} \FnRsumNq (1,\frac{1}{\rho c_k},{n \Nrb \NTx}) \hspace{0.2cm} \textrm{: Best-$1$ \FB}
                    \end{cases}.
    \end{equation}
}\TwoColumn
{
    \begin{equation}
        \label{eTwoSpecialSumRateTas}
        \hspace{-0.0cm} \Rsum =   \begin{cases}
                        \frac{1}{\Nus} \sum_{k=1}^{\Nus} \FnRsumNq (1,\frac{1}{\rho c_k},{ \Nus \NTx}) \hspace{1.5cm} \textrm{: Full \FB}\\
                        \frac{1}{\Nus} \sum_{k=1}^{\Nus} \sum_{n=1}^{\Nus} \binom{\Nus}{n} \left( \frac{1}{\Nrb} \right)^n \left( 1 - \frac{1}{\Nrb} \right)^{\Nus-n} \\
                        \hspace{1.0cm} \times \FnRsumNq (1,\frac{1}{\rho c_k},{n \Nrb \NTx}) \hspace{1.5cm} \textrm{: Best-$1$ \FB}.
                    \end{cases}
    \end{equation}
}

\subsection{Sum rate analysis for quantized \CQI}
\label{sRsumQ4Tas}
In this subsection, we provide the sum rate for the partial feedback \TAS{}-system with quantized \CQI.

\subsubsection{Feedback procedure and scheduling for the quantized system}
\label{sOverallProc4Q}
For quantization purposes, it is useful to work with normalized \CQI{} defined as $\Wkr = \tfrac{\Zkr}{c_k}$.
Each user computes $\Wkr$ for all the resource blocks and finds the best-$\NFb$ $\Wkr$'s.
Then, each user quantizes the selected $\Wkr$ values using a quantization policy depicted in \Fig{fQReg}.
In the figure, $J_\ell$ for $0 \leq \ell \leq L$ denotes the quantization region index and $\xi_\ell$ denotes the boundary value between regions.
More specifically, quantization is done as follows:
    \begin{equation}
        \label{eQuanPolicy}
        q_{k,r} = Q(\Wkr) = J_\ell, \hspace{0.4cm} \textrm{if} \hspace{0.4cm} \xi_\ell \leq \Wkr < \xi_{\ell+1}.
    \end{equation}
Then, each user feeds back the quantized region indices $q_{k,r}$'s for the selected best-$\NFb$ blocks to the transmitter together with the corresponding resource block indices.
To exploit multiuser diversity as in the non-quantized \CQI{} case, we assume for the scheduling policy that the transmitter schedules a transmission for each block to a user with the largest quantization region index.
When multiple users provide the same quantization index, the transmitter randomly selects a user.\OneColumn
{
    \begin{figure}
        \centering
            \includegraphics[width=3.0in]{\PfSchFigDir/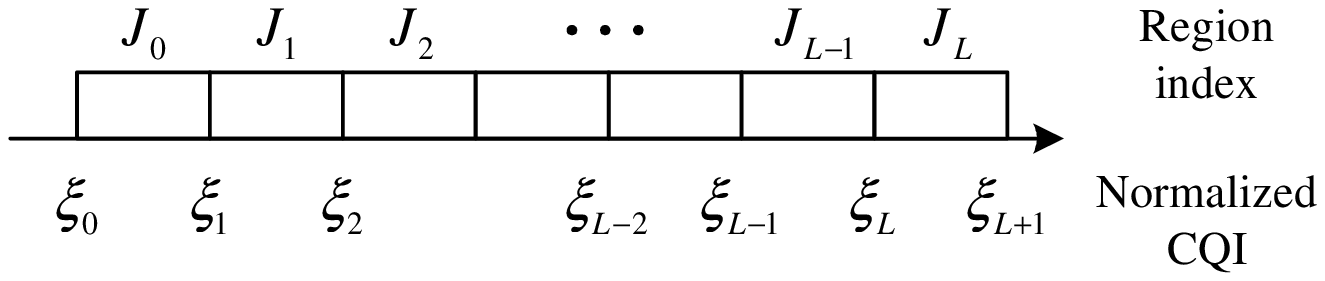}
            \caption{Quantization region for normalized \CQI. $(\xi_0 = 0, \; \xi_{L+1}= \infty)$}
            \label{fQReg}
    \end{figure}
}\TwoColumn
{
    \begin{figure}
        \centering
            \includegraphics[width=7.5cm]{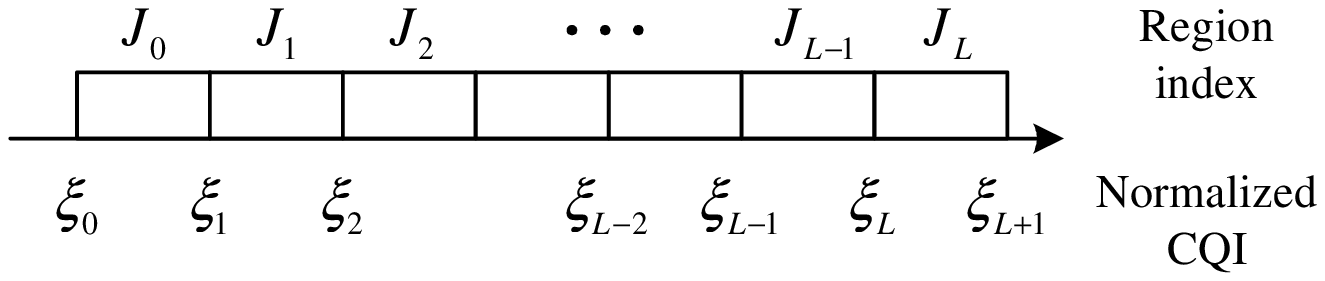}
            \caption{Quantization region for normalized \CQI. $(\xi_0 = 0, \; \xi_{L+1}= \infty)$}
            \label{fQReg}
    \end{figure}
}

\subsubsection{\Fstep-$1$, finding $\Fw(x)$}
\label{sStepOne4Q}
The step is related to determining the distribution of normalized \CQI{}.
Since normalized \CQI{} is $\Wkr = \tfrac{\Zkr}{c_k}$ and $\Pr\{ \Wkr \leq x \} = \Pr\{ \Zkr \leq {c_k} x \} $, the \CDF{} of $\Wkr$ with the \TAS{} diversity scheme is given from \eqref{eFzkTas} by
    \begin{equation}
        \label{eFw4TasQ}
        \Fw(x) = \Fzk (c_k x) = \{ \GammaTilde(1,x) \}^\NTx.
    \end{equation}

\subsubsection{\Fstep-$2$, finding $\Fu(x)$}
\label{sStepTwo4Q}
The step is related to the feedback policy and involves determining the order statistics for normalized \CQI{}, $W_{k,(1)} \leq \cdots \leq W_{k,(\Nrb)}$, quantizing $W_{k,(\ell)}$ for $\Nrb - \NFb + 1 \leq \ell \leq \Nrb$, and sending back the corresponding quantized region indices together with block indices.
Defining $\Ukr = \frac{\Ykr}{\rho c_k}$, from \Section{sStepThree4Nq} it denotes normalized \CQI{} as seen by the transmitter.
Since $\Pr\{ \Ukr \leq x \} = \Pr\{ \Ykr \leq {\rho c_k} x \}$, the \CDF{} of $\Ukr$ in \TAS{} is given by\OneColumn
{
    \begin{equation}
        \label{eUk4Tas}
            \Fu(x)=  \Fyk (\rho c_k x)  \stackrel{(a)}{=} \sum_{m=0}^{\NFb-1} e_1(\Nrb,\NFb,m) \{ \GammaTilde(1,x) \}^{(\Nrb-m) \NTx}
    \end{equation}
}\TwoColumn
{
    \begin{equation}
        \label{eUk4Tas}
            \Fu(x)=  \Fyk (\rho c_k x)  \hspace{-0.1cm} \stackrel{(a)}{=} \hspace{-0.1cm} \sum_{m=0}^{\NFb-1} \hspace{-0.2cm} e_1(\Nrb,\NFb,m) \{ \GammaTilde(1,x) \}^{(\Nrb-m) \NTx}
    \end{equation}
}
where $(a)$ follows from \eqref{eFy00} and \eqref{eFw4TasQ}, and $e_1(\Nrb,\NFb,m)$ is given in \eqref{eE1abm}.\footnote{$\Fu(x) \triangleq F_{{}_{\Ukr}}(x)$ for notational simplicity since $\Ukr$'s are \iid{} in $k$ and $r$.}
For the two special cases, we have $\Fu (x)= \{ \GammaTilde(1,x) \}^\NTx$ for full feedback $(\NFb=\Nrb)$ from \Corollary{cE1} and $\Fu (x)= \{ \GammaTilde(1,x) \}^{\Nrb \NTx}$ for best-$1$ feedback $(\NFb=1)$.

Let $\UkrQuan$ denote the quantization index received at the transmitter through feedback, which is equivalent to quantizing $\Ukr$ based on the policy in \eqref{eQuanPolicy}.
The distribution of $\UkrQuan$ can be readily determined from the distribution of $\Ukr$ given above.
It is shown in \Append{sStatOfUkrQuan} that $\UkrQuan$ is \iid{} in $k$ and $r$.
Then, a user with the largest $\UkrQuan$ is selected for block-$r$ by the transmitter in the next step.

\subsubsection{\Fstep-$3$, finding the conditional \PMF{} of $\XrQuan$}
\label{sStepThree4Q}
Let $\XrQuan$ denote the \SNR{} of a user selected for a transmission in block-$r$.
Suppose that $n$ users provided the quantization index for block-$r$, \ie{} $|S_r|=n$ recalling that $S_r$ denotes the set of those users.
We note that the probability for each user to be selected is equal since $\UkrQuan${}'s are \iid{} across users.
For the selected quantization index to be $J_\ell$, no one should provide a larger quantization index than $J_\ell$ (\ie{} $\UkrQuan \leq J_\ell$) and at least one user should provide the quantization index equal to $J_\ell$.\OneColumn
{
Thus, it is shown in \Append{sDerivePmf} that the conditional \PMF{} of $\XrQuan$ is given by
    \begin{equation}
        \label{eCondPmf}
        \hspace{-2.0cm} \Pr\{ \XrQuan = \rho c_k \xi_\ell \hspace{0.1cm} | \hspace{0.1cm} |S_r|= n \} = \tfrac{1}{\Nus} \left[ \left\{ \Fu(\xi_{\ell+1}) \right\}^n -  \left\{\Fu(\xi_\ell) \right\}^n \right], \hspace{0.2cm} 1 \leq k \leq \Nus, \hspace{0.2cm} 0 \leq \ell \leq L.
    \end{equation}
}\TwoColumn
{
Thus, for $1 \leq k \leq \Nus$ and $0 \leq \ell \leq L$, it is shown in \Append{sDerivePmf} that the conditional \PMF{} of $\XrQuan$ is given by
    \begin{equation}
        \label{eCondPmf}
        \hspace{-0.0cm} \Pr\{ \XrQuan = \rho c_k \xi_\ell \hspace{0.0cm} \big| \hspace{0.0cm} |S_r|= n \} = \tfrac{1}{\Nus} \left[ \left\{ \Fu(\xi_{\ell+1}) \right\}^n -  \left\{\Fu(\xi_\ell) \right\}^n \right].
    \end{equation}
}

\subsubsection{\Fstep-$4$, finding the sum rate}
\label{sStepFour4Q}
To calculate the sum rate, we assume that the modulation level for the transmission to the selected user-$k$ is assumed to be determined as $\log( 1 + \rho c_k \xi_\ell)$ so as to prevent an outage of the link when user-$k$ with a quantization level $J_\ell$ is selected.
It is shown in \Append{sDeriveRsumQ} that the sum rate is given by\OneColumn
{
    \begin{equation}
        \label{eRsumQ}
        \hspace{-0.5cm} \Rsum = \Ex[ \log(1+\XrQuan)] = \sum_{k=1}^\Nus \sum_{\ell=1}^L \frac{\log_2 (1 + \rho c_k \xi_\ell)}{\Nus} \times \FnRsumQ \left(\Fu(\xi_{\ell}),\Fu(\xi_{\ell+1}),\Nus,\frac{\NFb}{\Nrb} \right),
    \end{equation}
}\TwoColumn
{
    \begin{equation*}
        \hspace{-0.5cm} \Rsum = \Ex[ \log(1+\XrQuan)] = \sum_{k=1}^\Nus \sum_{\ell=1}^L \frac{\log_2 (1 + \rho c_k \xi_\ell)}{\Nus}
    \end{equation*}
    \begin{equation}
        \label{eRsumQ}
        \hspace{2.5cm} \times \FnRsumQ \left(\Fu(\xi_{\ell}),\Fu(\xi_{\ell+1}),\Nus,\frac{\NFb}{\Nrb} \right),
    \end{equation}
}
where $\FnRsumQ (x,y,z,r)$ is given by\OneColumn
{
    \begin{equation}
        \label{eC3}
        \hspace{-2.5cm} \FnRsumQ (x,y,z,r)= \left\{ 1 - r \left( 1 - y \right) \right \}^z - \left\{ 1 - r \left( 1 - x \right) \right \}^z.
    \end{equation}
}\TwoColumn
{
    \begin{equation}
        \label{eC3}
        \hspace{-0.0cm} \FnRsumQ (x,y,z,r)= \left\{ 1 - r \left( 1 - y \right) \right \}^z - \left\{ 1 - r \left( 1 - x \right) \right \}^z.
    \end{equation}
}
For full feedback as a special case, we have\OneColumn
{
    \begin{equation}
        \label{eRsumQFullFb}
        \hspace{-1.7cm} \Rsum= \tfrac{1}{\Nus} \sum_{k=1}^{\Nus} \sum_{\ell=0}^{L} \log(1+ \rho c_k \xi_\ell) \left[ \left\{ \Fw(\xi_{\ell+1}) \right\}^\Nus -  \left\{\Fw(\xi_\ell) \right\}^\Nus \right].
    \end{equation}
}\TwoColumn
{
    \begin{equation*}
        \hspace{-3.0cm} \Rsum= \tfrac{1}{\Nus} \sum_{k=1}^{\Nus} \sum_{\ell=0}^{L} \log(1+ \rho c_k \xi_\ell) \end{equation*}
    \begin{equation}
        \label{eRsumQFullFb}
        \hspace{-0.0cm} \times \left[ \left\{ \Fw(\xi_{\ell+1}) \right\}^\Nus -  \left\{\Fw(\xi_\ell) \right\}^\Nus \right].
    \end{equation}
}

\section{Sum rate analysis with application to \OSTBC{} and \CDD{}}
\label{sPerformanceAnalysis4OstbcCdd}
Since the diversity technique affects the distribution of $\Zkr$ or $\Wkr$ in \Fstep-$1$, we focus in this section on deriving $\Fzk$ and $\Fw$ for \OSTBC{} and \CDD{}.
\Fstep-$2$ and \Fstep-$3$ from the TAS analysis can be adopted with no change.
Then, we can obtain the sum rate by carrying out \Fstep-$4$.

\subsection{Sum rate for the orthogonal space time block codes (\OSTBC) scheme}
\subsubsection{Sum rate for non-quantized \CQI{} feedback}
\label{sSnrDistribution4OstbcCdd}
For the equal power transmission from each antenna in \OSTBC{}, effective \CQI{} of user-$k$ at block-$r$ is given by the square of the $2$-norm of a channel vector from the transmit antennas normalized by the number of transmit antennas \cite{bAlamouti98,bJorswieckEurasip09}, \ie{}
    \begin{equation}
        \label{eCqiOstbc}
        \hspace{-0.5cm} \Zkr= |H_{k,r}|^2 = \frac{1}{\NTx} \sum_{i=1}^{\NTx} |H_{k,r,i}|^2.
    \end{equation}
Since we assume that $H_{k,r,i}$ follows $\CN(0,c_k)$, $|H_{k,r,i}|^2$ follows the Gamma distribution $\Gcal(1,\frac{1}{c_k})$ \cite[(17.6)]{bJohnsonWiley94}.
The sum of $n$ \iid{} random variables with $\Gcal(\alpha, \beta)$ follows the Gamma distribution $\Gcal(n \alpha, \beta)$ \cite[2-1-110]{bProakisMgh95} and a Gamma distributed random variable with $\Gcal(\alpha,\beta)$ multiplied by a constant $c$ follows the distribution of $\Gcal(\alpha, \frac{\beta}{c})$.\footnote{For $Y = c X$ where $X$ follows $\Gcal(\alpha,\beta)$, since $\Fx(x) = \GammaTilde(\alpha,\beta x)$ from \eqref{eGammaCdf}, $\Fy(x) = \Pr\{ c X \leq x \} = \Pr\{ X \leq \tfrac{x}{c} \} = \Fx(\tfrac{x}{c}) = \GammaTilde(\alpha,\tfrac{\beta x}{c})$, which means that $Y$ follows $\Gcal(\alpha,\tfrac{\beta}{c})$.}
Therefore, \CQI{} $\Zkr$ in \eqref{eCqiOstbc} follows the Gamma distribution with $\Gcal(\NTx, \frac{\NTx}{c_k})$.
Thus, the \CDF{} of $\Zkr$ for \Fstep-$1$ is given from \eqref{eGammaCdf} by
    \begin{equation}
        \label{eFzkOstbc}
        \hspace{-1.0cm} \Fzk(x) = \GammaTilde \big( \NTx,\tfrac{\NTx x}{c_k} \big).
    \end{equation}

Since the feedback policy and the scheduling policy in \OSTBC{} are the same as in \TAS{}, we can follow the same next two steps, specifically \Fstep-$2$ in \Section{sStepTwo4Nq} and \Fstep-$3$ in \Section{sStepThree4Nq}.
Then, we obtain the conditional \CDF{} of $X_r$, the \SNR{} of a selected user in block-$r$, which is given for the general case in \eqref{eCdfX} and for two special cases in \eqref{eTwoSpecialCdfs} where $\Fzk(x)$ is to be replaced by \eqref{eFzkOstbc}.

We can carry out \Fstep-$4$ by again exploiting the fact that the conditional \CDF{} in \eqref{eCdfX} is represented in terms of a polynomial in $\Fzk(x)$ in \eqref{eFzkOstbc} and using the integration identity in \eqref{eSumRateInt}.
The sum rate $\Ex[\log(1+X_r)]$ of \OSTBC{} for the general case of $\NFb$ can be shown to be given by\OneColumn
{
    \begin{equation}
        \label{eSumRateGenFbOstbc}
        \hspace{-0.0cm} \Rsum = \tfrac{1}{\Nus} \sum_{k=1}^{\Nus} \sum_{n=1}^{\Nus} \tbinom{\Nus}{n} \left( \tfrac{\NFb}{\Nrb} \right)^n \left( 1 - \tfrac{\NFb}{\Nrb} \right)^{\Nus-n} \; \sum_{m=0}^{n(\NFb-1)} e_2(\Nrb,\NFb,n,m) \FnRsumNq (\NTx,\tfrac{\NTx}{\rho c_k},{n \Nrb - m}).
    \end{equation}
}\TwoColumn
{
    \begin{equation*}
        \hspace{-1.0cm} \Rsum = \tfrac{1}{\Nus} \sum_{k=1}^{\Nus} \sum_{n=1}^{\Nus} \tbinom{\Nus}{n} \left( \tfrac{\NFb}{\Nrb} \right)^n \left( 1 - \tfrac{\NFb}{\Nrb} \right)^{\Nus-n}
    \end{equation*}
    \begin{equation}
        \label{eSumRateGenFbOstbc}
        \hspace{0.0cm} \times \sum_{m=0}^{n(\NFb-1)} e_2(\Nrb,\NFb,n,m) \FnRsumNq (\NTx,\tfrac{\NTx}{\rho c_k},{n \Nrb - m}).
    \end{equation}
}
From \eqref{eTwoSpecialCdfs} and \eqref{eSumRateInt}, we have the sum rate for two special cases (\ie{} $\NFb= \Nrb$ and $\NFb=1$) as\OneColumn
{
    \begin{equation}
        \label{eTwoSpecialSumRateOstbc}
        \hspace{-0.0cm} \Rsum =   \begin{cases}
                        \frac{1}{\Nus} \sum_{k=1}^{\Nus} \FnRsumNq (\NTx,\frac{\NTx}{\rho c_k},{\Nus}) \hspace{6.3cm} \textrm{: Full \FB}\\
                        \frac{1}{\Nus} \sum_{k=1}^{\Nus} \sum_{n=1}^{\Nus} \binom{\Nus}{n} \left( \frac{1}{\Nrb} \right)^n \left( 1 - \frac{1}{\Nrb} \right)^{\Nus-n} \FnRsumNq (\NTx,\frac{\NTx}{\rho c_k},{n \Nrb}) \hspace{0.2cm} \textrm{: Best-$1$ \FB}
                    \end{cases}.
    \end{equation}
}\TwoColumn
{
    \begin{equation}
        \label{eTwoSpecialSumRateOstbc}
        \hspace{-0.0cm} \Rsum =   \begin{cases}
                        \frac{1}{\Nus} \sum_{k=1}^{\Nus} \FnRsumNq (\NTx,\frac{\NTx}{\rho c_k},{\Nus}) \hspace{2.0cm} \textrm{: Full \FB}\\
                        \frac{1}{\Nus} \sum_{k=1}^{\Nus} \sum_{n=1}^{\Nus} \binom{\Nus}{n} \left( \frac{1}{\Nrb} \right)^n \left( 1 - \frac{1}{\Nrb} \right)^{\Nus-n} \\
                        \hspace{1.0cm} \times \FnRsumNq (\NTx,\frac{\NTx}{\rho c_k},{n \Nrb}) \hspace{1.5cm} \textrm{: Best-$1$ \FB}.
                    \end{cases}
    \end{equation}
}

Since the maximum code rate for complex \OSTBC{} is 1 only for $\NTx=2$ and less than $1$ otherwise \cite{bJafarkhaniAp05}, we note that the exact sum rate can be obtained by multiplying the code rate, \ie{} multiplying $\frac{3}{4}$ for $\NTx=3$ and $4$.

\subsubsection{Sum rate for quantized \CQI{} feedback}
\label{sRsumQ4Ostbc}
We consider the same policy for quantization, feedback, and scheduling as that in \Section{sOverallProc4Q}.
Since normalized \CQI{} is $\Wkr = \tfrac{\Zkr}{c_k}$ and $\Pr\{ \Wkr \leq x \} = \Pr\{ \Zkr \leq {c_k} x \} $, the \CDF{} of $\Wkr$ in \OSTBC{} for \Fstep-$1$ is given from \eqref{eFzkOstbc} by
    \begin{equation}
        \label{eFw4OstbcQ}
        \Fw(x) = \Fzk (c_k x) = \GammaTilde(\NTx,\NTx x).
    \end{equation}
Normalized \CQI{} viewed at the transmitter for user-$k$ at block-$r$ is $\Ukr = \frac{\Ykr}{\rho c_k}$.
As in \Fstep-$2$ in \Section{sStepTwo4Q}, the \CDF{} of $\Ukr$ for \OSTBC{} is given by\OneColumn
{
    \begin{equation}
        \label{eUk4Ostbc}
            \Fu(x)= \Fyk (\rho c_k x)  \stackrel{(a)}{=} \sum_{m=0}^{\NFb-1} e_1(\Nrb,\NFb,m) \{ \GammaTilde(\NTx,\NTx x) \}^{\Nrb-m}
    \end{equation}
}\TwoColumn
{
    \begin{equation}
        \label{eUk4Ostbc}
            \Fu(x)= \Fyk (\rho c_k x) \hspace{-0.1cm} \stackrel{(a)}{=} \hspace{-0.1cm} \sum_{m=0}^{\NFb-1} \hspace{-0.2cm} e_1(\Nrb,\NFb,m) \{ \GammaTilde(\NTx,\NTx x) \}^{\Nrb-m}
    \end{equation}
}
where $(a)$ follows from \eqref{eFy00} and \eqref{eFw4OstbcQ}, and $e_1(\Nrb,\NFb,m)$ is given in \eqref{eE1abm}.
For the two special cases, we have $\Fu (x) = \GammaTilde(\NTx,\NTx x)$ for full feedback $(\NFb=\Nrb)$, and $\{ \GammaTilde(\NTx,\NTx x) \}^\Nrb $ for best-$1$ feedback $(\NFb=1)$.
Since the conditional \PMF{} of the \SNR{} for a selected user for \Fstep-$3$ is the same as \eqref{eCondPmf}, the sum rate of \OSTBC{} has the same form as \eqref{eRsumQ} where $\Fu(x)$ in \eqref{eUk4Ostbc} is to be substituted.

\subsection{Sum rate for the cyclic delay diversity (\CDD)}
\subsubsection{Sum rate for non-quantized \CQI{} feedback}
As in \OSTBC{} and \TAS{}, we derive the sum  rate for \CDD{} by first obtaining the \CDF{} of $\Zkr$ for \Fstep-$1$ and then using the same remaining 3-steps of the framework in \Table{tMainStepsForRsumAnal}.
For equal power transmission from each antenna, the equivalent channel of \CDD{} with cyclic delay $D_i$ at each transmit antenna is a dot product of a channel vector and complex phases determined by the cyclic delays \cite{bDammannGcom01}, \ie{} $H_{k,r}= \frac{1}{\sqrt{\NTx}} \sum_{i=1}^{\NTx} H_{k,r,i} e^{j \frac{2 \pi}{N} D_i}$.
The resulting channel follows $\CN(0,c_k)$ since $H_{k,r}$ is a linear combination of complex Gaussian random variables \cite{bGarciaAw94}.
Thus, \CQI{} for the equivalent channel of user-$k$ at block-$r$ is given by
    \begin{equation}
        \label{eCqiCdd}
        \hspace{-0.5cm} \Zkr= |H_{k,r}|^2 = \frac{1}{\NTx} \Big| \sum_{i=1}^{\NTx} H_{k,r,i} e^{j \frac{2 \pi}{N} D_i} \Big|^2,
    \end{equation}
which follows the Gamma distribution with $\Gcal \big( 1, \tfrac{1}{c_k} \big) $ \cite[(17.6)]{bJohnsonWiley94}.
From \eqref{eGammaCdf}, the \CDF{} of $\Zkr$ for \Fstep-$1$ is given by
    \begin{equation}
        \label{eFzkCdd}
        \hspace{-0.0cm} \Fzk(x) = \GammaTilde \left ( 1,\tfrac{x}{c_k} \right).
    \end{equation}
We can see that $\Fzk(x)$ in \eqref{eFzkCdd} for \CDD{} is the same as that in \eqref{eFzkTas} for \TAS{} and in \eqref{eFzkOstbc} for \OSTBC{} where $\NTx=1$.
Thus, the sum rate of \CDD{} is exactly the same as that in \eqref{eSumRateGenFbTas} and \eqref{eTwoSpecialSumRateTas} for \TAS{} and in \eqref{eSumRateGenFbOstbc} and \eqref{eTwoSpecialSumRateOstbc} for \OSTBC{} where $\NTx=1$.

We note in \cite{bDammannGcom01,bViswanathTit02} that \CDD{} or opportunistic beamforming is a technique to enhance the frequency diversity in a given channel by multiplying a gain to the channel randomly but in a controlled manner.
We also note that the diversity gain increases with the number of the transmit antennas.
However, since blocks are assumed to be already independent in our channel model, \CDD{} does not have a room to increase frequency diversity even though we increase the number of the transmit antennas.
Thus, we verify that the distribution of \CQI{} of \CDD{} in \eqref{eFzkCdd} does not depend on $\NTx$.

\subsubsection{Sum rate for quantized \CQI{} feedback}
Since normalized \CQI{} is $\Wkr = \tfrac{\Zkr}{c_k}$ and $\Pr\{ \Wkr \leq x \} = \Pr\{ \Zkr \leq {c_k} x \} $, the \CDF{} of $\Wkr$ in \CDD{} for \Fstep-$1$ is given from \eqref{eFzkCdd} by
    \begin{equation}
        \label{eFw4CddQ}
        \Fw(x) = \Fzk (c_k x) = \GammaTilde(1,x).
    \end{equation}
Normalized \CQI{} viewed at the transmitter for user-$k$ at block-$r$ is $\Ukr = \frac{\Ykr}{\rho c_k}$.
Through the same step as \Fstep-$2$ in \Section{sStepTwo4Q}, the \CDF{} of $\Ukr$ for \Fstep-$2$ is given by
    \begin{equation}
        \label{eUk4Cdd}
            \Fu(x)= \Fyk (\rho c_k x)  \stackrel{(a)}{=} \sum_{m=0}^{\NFb-1} e_1(\Nrb,\NFb,m) \{ \GammaTilde(1,x) \}^{\Nrb-m}
    \end{equation}
where $(a)$ follows from \eqref{eFy00} and \eqref{eFw4CddQ}, and $e_1(\Nrb,\NFb,m)$ is given in \eqref{eE1abm}.
Since the conditional \PMF{} of the \SNR{} for a selected user for \Fstep-$3$ is the same as \eqref{eCondPmf}, the sum rate of \CDD{} is given by \eqref{eRsumQ} with $\Fu(x)$ in \eqref{eUk4Cdd}.
We can verify that the sum rate of \CDD{} does not depend on $\NTx$ since blocks are assumed to be independent.

\section{Relation between Probability of normal scheduling and the sum rate ratio}
\label{sRelation}
In this section, we investigate the problem of minimizing the amount of feedback in the system by examining how much feedback is required to maintain the sum rate comparable to the sum rate obtained by a full feedback scheme.
Let $\RFb = \frac{\NFb}{\Nrb}$ denote the feedback ratio, \ie{} the ratio of the number of feedback blocks to the total number of blocks.
The design objective is to find the minimum feedback ratio while the achieved sum rate is above a certain fraction of the sum rate obtained by a full feedback scheme, \ie{}\OneColumn
{
    \begin{equation}
        \label{eReqRfb}
        \textrm{Find the minimum } \RFb, \hspace{0.2cm} \st{} \hspace{0.2cm} \RsumRatio = \frac{\Rsum \textrm{ by partial feedback}}{\Rsum \textrm{ by full feedback}} \geq \eta.
    \end{equation}
}\TwoColumn
{
    \begin{equation}
        \label{eReqRfb}
        \textrm{Find min. } \RFb, \hspace{0.1cm} \st{} \hspace{0.2cm} \RsumRatio = \frac{\Rsum \textrm{ by partial feedback}}{\Rsum \textrm{ by full feedback}} \geq \eta.
    \end{equation}
}
Since we have the expressions for the sum rate for both partial and full feedback schemes, they can be substituted in the above equation and one can solve for $\NFb$.
Here we make two simplifications and obtain a more tractable expression.
We carry this out for the \OSTBC{} diversity scheme.

First we note from \eqref{eCdfX} that we have
    \begin{equation}
        \label{eE2Characteristics}
        \sum_{m=0}^{n(\NFb-1)} e_2(\Nrb,\NFb,n,m) =1,
    \end{equation}
since $\FxCondSr (\infty)=1$ and $\Fzk(\infty)=1$ by the \CDF{} property \cite{bGarciaAw94}.
Second we note that $\FnRsumNq (x,y,z)$ in \eqref{eC1Ori} has almost the same value for large $z$ when $x$ and $y$ are fixed.
This is graphically illustrated  in \Fig{fCompI1}.
We assume that $\FnRsumNq (x,y,z_1) \simeq \FnRsumNq (x,y,z_2)$ for large $z_1$ and $z_2$.
More specifically, when we assume that $\FnRsumNq (\NTx,\frac{\NTx}{\rho c_k},n \Nrb - m) \simeq \FnRsumNq (\NTx,\frac{\NTx}{\rho c_k},\Nus)$ in \eqref{eSumRateGenFbOstbc} and using \eqref{eE2Characteristics}, the sum rate of \OSTBC{} for partial feedback in \eqref{eSumRateGenFbOstbc} reduces to\OneColumn
{
    \begin{equation*}
        \hspace{-0.0cm} \Rsum \simeq \tfrac{1}{\Nus} \sum_{k=1}^{\Nus} \FnRsumNq (\NTx,\tfrac{\NTx}{\rho c_k},\Nus) \sum_{n=1}^{\Nus} \tbinom{\Nus}{n} \left( \tfrac{\NFb}{\Nrb} \right)^n \left( 1 - \tfrac{\NFb}{\Nrb} \right)^{\Nus-n}
    \end{equation*}
    \begin{equation}
        \label{eSumRateApprOstbc}
        \hspace{-1.7cm} \stackrel{(a)}{=} \tfrac{1}{\Nus} \sum_{k=1}^{\Nus} \FnRsumNq (\NTx,\tfrac{\NTx}{\rho c_k},\Nus) \; \left( 1 - ( 1 - \tfrac{\NFb}{\Nrb})^\Nus \right),
    \end{equation}
}\TwoColumn
{
    \begin{equation*}
        \hspace{-0.0cm} \Rsum \simeq \sum_{k=1}^{\Nus} \tfrac{\FnRsumNq (\NTx,\tfrac{\NTx}{\rho c_k},\Nus)}{\Nus} \sum_{n=1}^{\Nus} \tbinom{\Nus}{n} \left( \tfrac{\NFb}{\Nrb} \right)^n \left( 1 - \tfrac{\NFb}{\Nrb} \right)^{\Nus-n}
    \end{equation*}
    \begin{equation}
        \label{eSumRateApprOstbc}
        \hspace{-1.7cm} \stackrel{(a)}{=} \sum_{k=1}^{\Nus} \tfrac{\FnRsumNq (\NTx,\tfrac{\NTx}{\rho c_k},\Nus)}{\Nus} \; \left( 1 - ( 1 - \tfrac{\NFb}{\Nrb})^\Nus \right),
    \end{equation}
}
where $(a)$ follows from the binomial theorem \cite{bGarciaAw94}.
From the sum rate obtained by a full feedback scheme in \eqref{eTwoSpecialSumRateOstbc} and the sum rate obtained by partial feedback in \eqref{eSumRateApprOstbc}, we have
    \begin{equation}
        \label{eSumRateRatio}
        \RsumRatio = \frac{\Rsum \textrm{ by partial feedback}}{\Rsum \textrm{ by full feedback}} \simeq 1 - ( 1 - \tfrac{\NFb}{\Nrb})^\Nus.
    \end{equation}
This approximation is well supported by the numerical results in \Section{sNumericalResults4RsumAnal}.
We note that the right-hand side in \eqref{eSumRateRatio} is exactly the same as the probability that at least one user provides \CQI{} to the transmitter in a block, \ie{} a probability of the complement of a scheduling outage.\footnote{Since a scheduling outage in a block happens when no user provides \CQI{} for that block, its probability is $( 1 - \tfrac{\NFb}{\Nrb})^\Nus$.}
From \eqref{eReqRfb} and \eqref{eSumRateRatio}, the required feedback ratio is given by
    \begin{equation}
        \label{eRfb}
        \RFb = \tfrac{\NFb}{\Nrb} \geq 1 - (1 - \eta)^{\frac{1}{\Nus}}.
    \end{equation}
We note here that the required feedback ratio does not depend on the number of antennas and user distribution in the average \SNR{} but on the number of users.
In the same way, we compute the required feedback ratio for \TAS{} making the same assumption about $\FnRsumNq (x,y,z)$ and obtain the same result as \eqref{eRfb}.
It is useful to note that the required feedback ratio in our system with fixed amount of feedback can be derived from the scheduling outage probability using the approximation above.
This has similarities to the problem of determining the required threshold in a threshold-based feedback system considering a scheduling outage as in \cite{bGesbertIcc04}.\OneColumn
{
   \begin{figure}
       \centering
           \includegraphics[width=12.0cm]{\PfSchFigDir/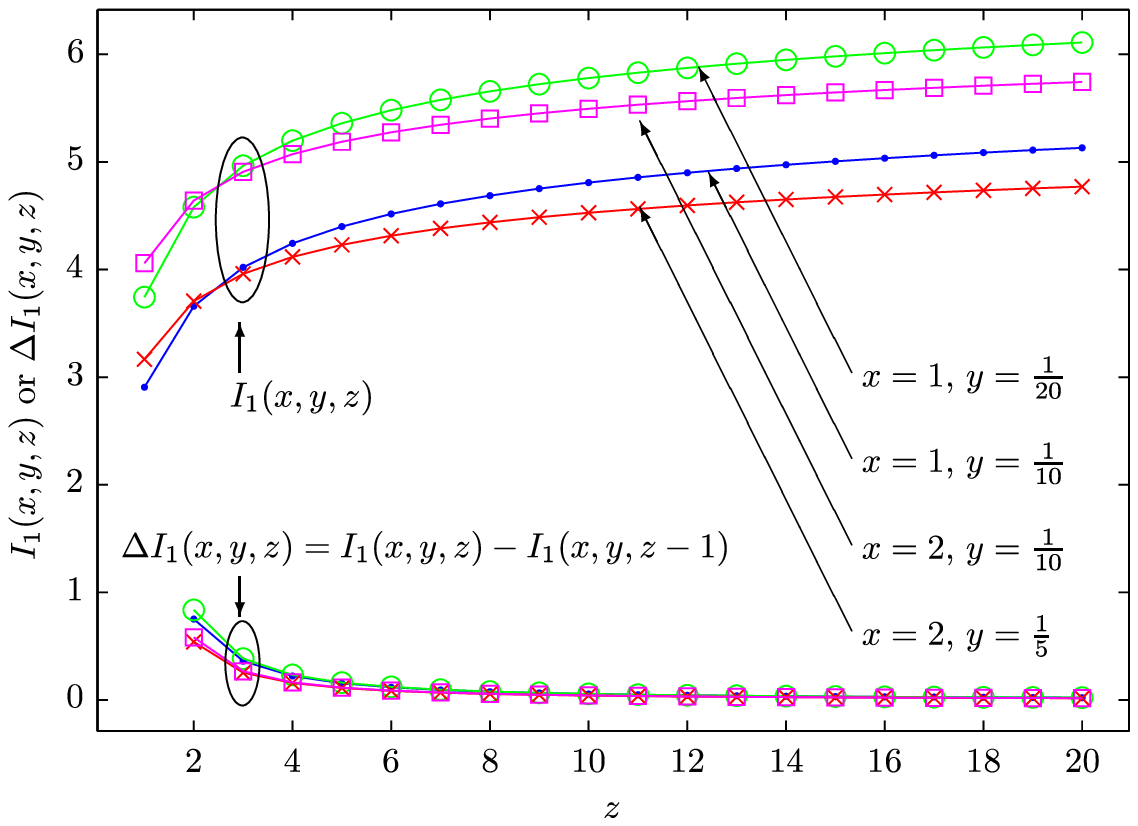}
           \caption{$\FnRsumNq (x,y,z)$ and its slope. We note that when $x$ and $y$ are fixed, the rate of increase in $\FnRsumNq (x,y,z)$ is very small when $z$ is large.}
           \label{fCompI1}
   \end{figure}
}\TwoColumn
{
   \begin{figure}
       \centering
           \includegraphics[width=8.0cm]{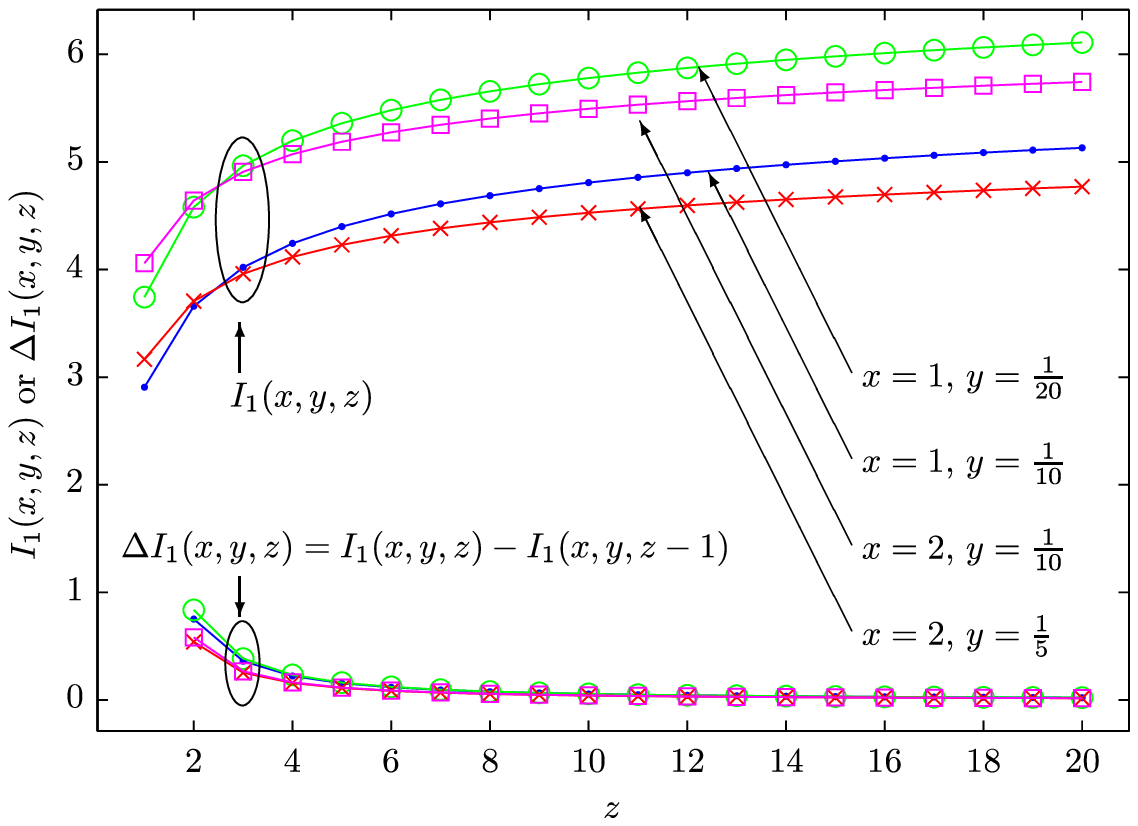}
           \caption{$\FnRsumNq (x,y,z)$ and its slope. We note that when $x$ and $y$ are fixed, the rate of increase in $\FnRsumNq (x,y,z)$ is very small when $z$ is large.}
           \label{fCompI1}
   \end{figure}
}

The above analysis was conducted assuming unquantized CQI.
A similar analysis can be carried out assuming quantized CQI and employing some approximations one can obtain the same result as \eqref{eSumRateRatio} and \eqref{eRfb}.
We omit the details.



\section{Numerical Results}
\label{sNumericalResults4RsumAnal}
In this section, we conduct a numerical study of the analytical results to obtain some insight.
To reflect asymmetrical user distribution in their average \SNR{}, we use the exponential decay model for the average channel power of users \cite{bJorswieckEurasip09}:
    \begin{equation}
        \label{eUserDist}
        c_k = c \; e^{- \lambda k}, \hspace{0.2cm} \st{} \hspace{0.2cm} \sum_{k=1}^{\Nus} c_k = \Nus.
    \end{equation}
We can see that $\lambda = 0$ corresponds to \iid{} users and that user asymmetry increases with $\lambda$.

\subsection{Effect of partial feedback on the sum rate}
In \Fig{fNfbEffect}, we show the sum rate results computed using the analytical expressions and the simulation results as a function of the number of users.
In the figure, \TAS{} with $\NTx=2$ is used and the average channel power is identical across users (\ie{} $\lambda=0$) in \Fig{fNfbEffect}(a) and different in \Fig{fNfbEffect}(b).
We can see that both analytic and simulation results are well matched.
We can also see the effect of the feedback ratio ($\RFb$) on the sum rate.
As we expect, the sum rate increases with the feedback ratio for both choices of $\lambda$.
We note that the throughput gap between best-$1$ feedback ($\RFb=0.1$) and full feedback ($\RFb=1.0$) is large even when the number of users is 20.
When the number of users is smaller than 10, we need $\RFb \geq 0.4$ to attain a throughput comparable to a full feedback scheme.
For $\lambda \neq 0$, we also note that fairness provided by proportional fair scheduling decreases the sum rate when the number of users is large, because the throughput variation is larger in the larger population and the throughput function $(\log_2(1+x))$ is concave, which is known as the fairness-capacity trade-off in \cite{bJorswieckEurasip09}.\OneColumn
{
    \begin{figure}
        \centering
        \includegraphics[width=12.0cm]{\PfSchFigDir/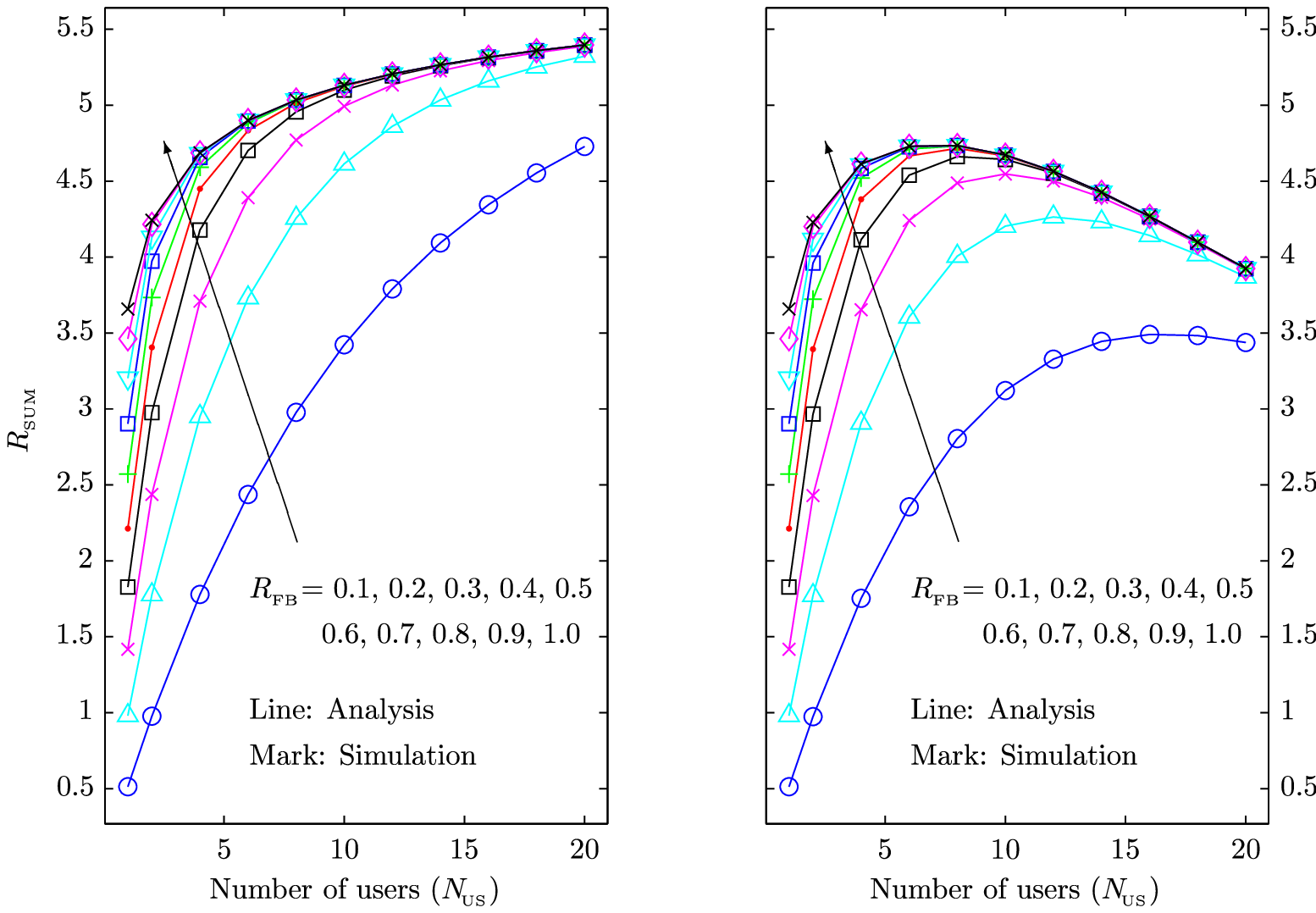}\\
        \begin{tabular}[c]{ll}
            \vspace{-1.2cm} {} & {} \\
            {\scriptsize \hspace{-0.0cm} (a) $\lambda = 0$ \hspace{2.2cm}} & {\scriptsize \hspace{0.2cm} (b) $\lambda = 0.3$ \hspace{-0.0cm}}\\
        \end{tabular}
        \vspace{-0.5cm}
        \caption{Effect of feedback ratio $(\RFb = \frac{\NFb}{\Nrb})$ on the sum rate for different $\lambda$ in \eqref{eUserDist}. (\TAS{}, $\Nrb=10$, $\NTx=2$, and Tx \SNR{}= 10\dB)}
        \label{fNfbEffect}
    \end{figure}
}\TwoColumn
{
    \begin{figure}
        \centering
        \includegraphics[width=8.0cm]{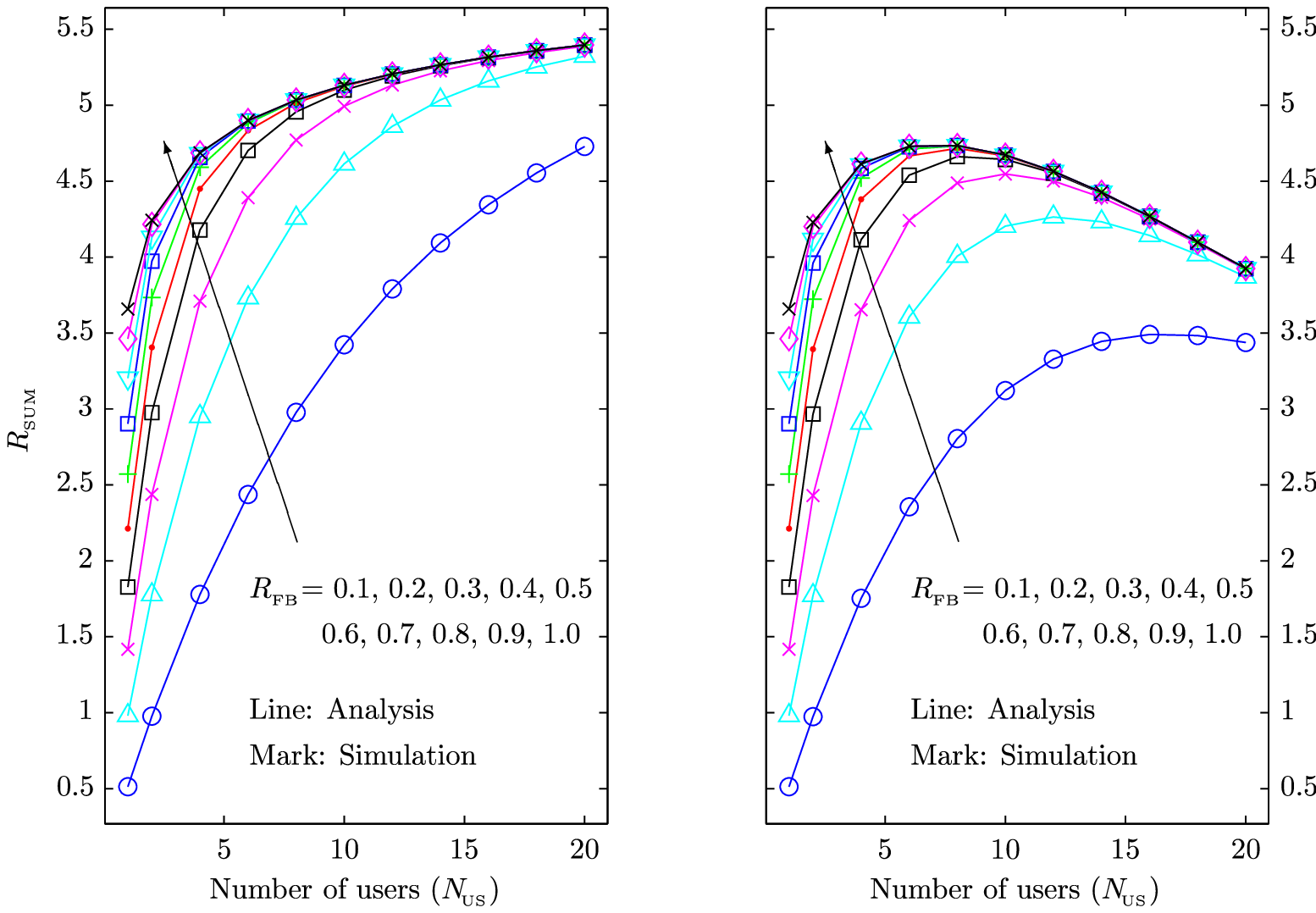}\\
        \begin{tabular}[c]{ll}
            \vspace{-1.2cm} {} & {} \\
            {\scriptsize \hspace{-0.0cm} (a) $\lambda = 0$ \hspace{2.2cm}} & {\scriptsize \hspace{0.2cm} (b) $\lambda = 0.3$ \hspace{-0.0cm}}\\
        \end{tabular}
        \vspace{-0.5cm}
        \caption{Effect of feedback ratio $(\RFb = \frac{\NFb}{\Nrb})$ on the sum rate for different $\lambda$ in \eqref{eUserDist}. (\TAS{}, $\Nrb=10$, $\NTx=2$, and Tx \SNR{}= 10\dB)}
        \label{fNfbEffect}
    \end{figure}
}

In \Fig{fNTxEffect}, we show the effect of the number of antennas for both \TAS{} and \OSTBC{} schemes with partial feedback.
Users are asymmetrically distributed $(\ie{}$ $\lambda \neq 0)$.
In general, multiuser diversity increases with the number of users, as well as the mean and the variance of the signal quality \cite{bHurPrep10}.
Since selection of antennas in \TAS{} can be regarded as an increase of the number of users due to the increase of candidate channels for the communication, the sum rate of \TAS{} increases with $\NTx$.
However, since \OSTBC{} decreases the variance of the signal quality by the averaging effect shown in \eqref{eCqiOstbc}, the sum rate of \OSTBC{} decreases with $\NTx$.
In both feedback ratio of $\RFb=0.1$ and $\RFb=0.5$, we can verify this effect of the number of antennas on the sum rate for each transmit antenna scheme.\OneColumn
{
   \begin{figure}
       \centering
           \includegraphics[width=12.0cm]{\PfSchFigDir/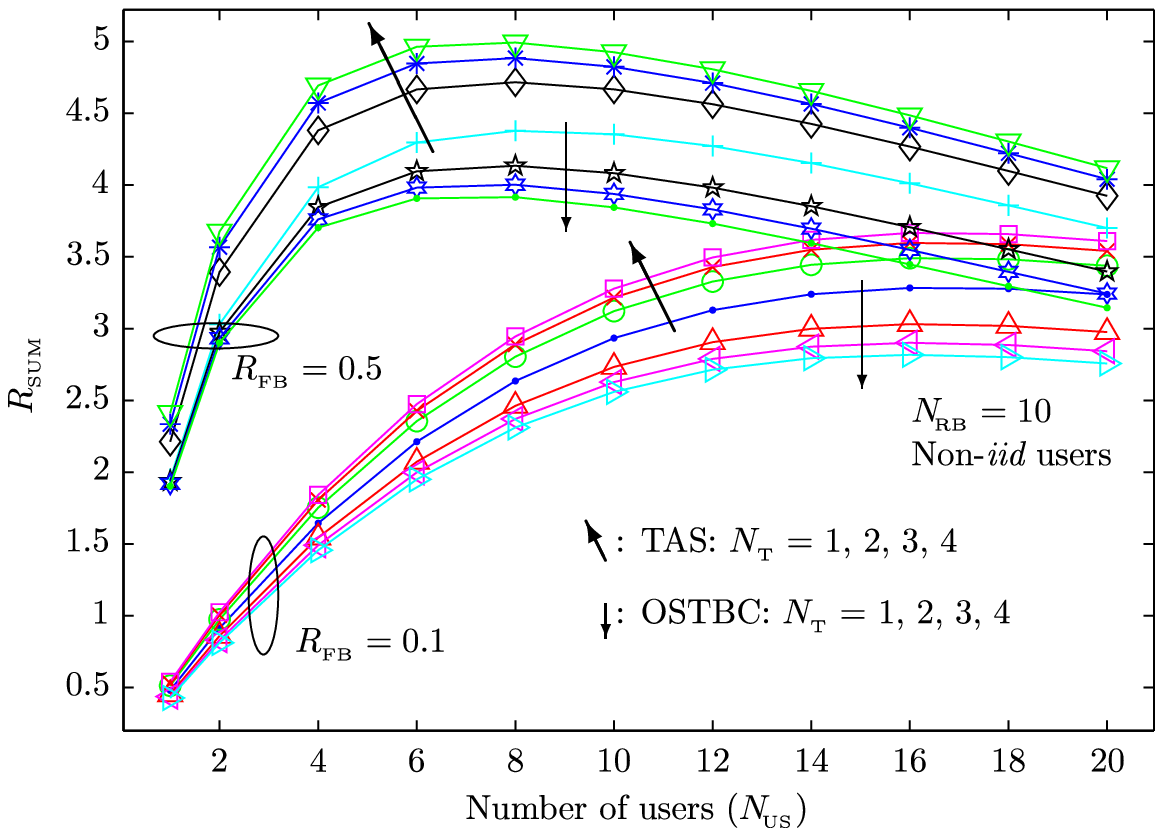}
           \caption{Effect of the number of antennas on the sum rate with partial feedback. (\TAS{} and \OSTBC, $\Nrb=10$, $\lambda$ in \eqref{eUserDist}= 0.3, and Tx \SNR{}= 10\dB.)}
           \label{fNTxEffect}
   \end{figure}
}\TwoColumn
{
   \begin{figure}
       \centering
           \includegraphics[width=8.0cm]{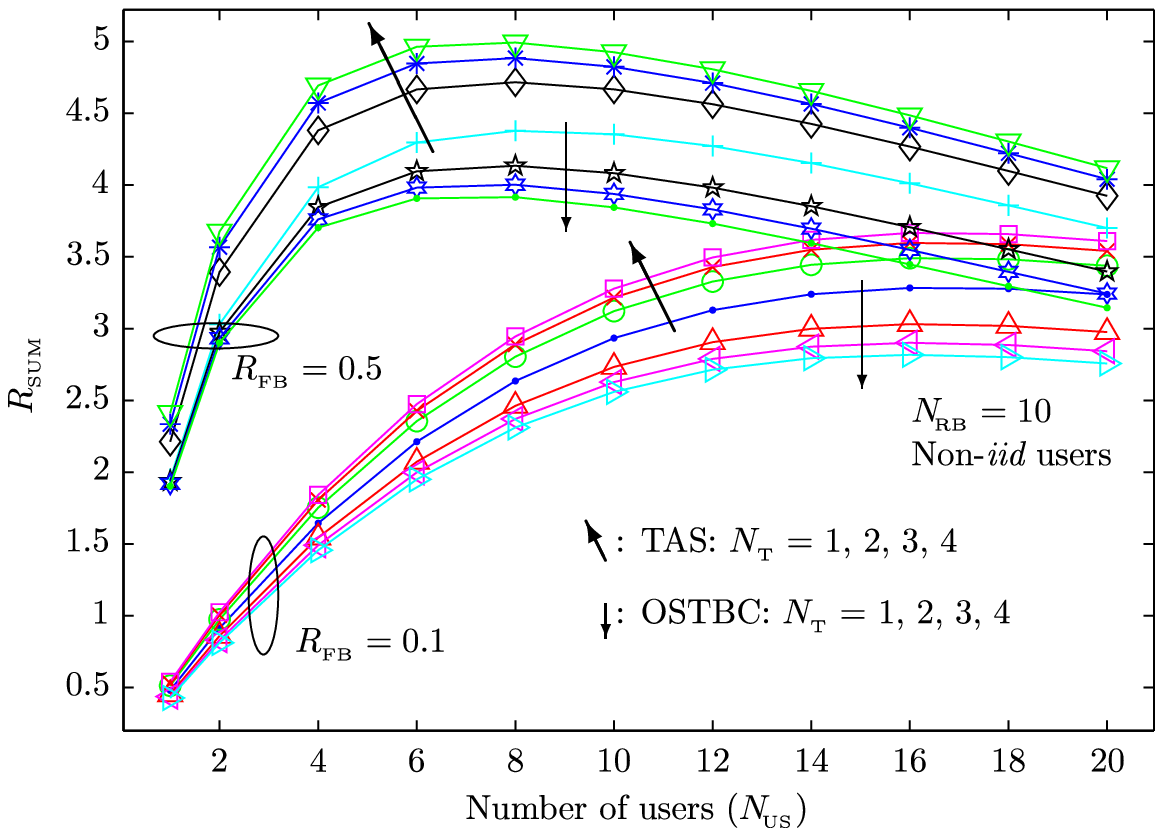}
           \caption{Effect of the number of antennas on the sum rate with partial feedback. (\TAS{} and \OSTBC, $\Nrb=10$, $\lambda$ in \eqref{eUserDist}= 0.3, and Tx \SNR{}= 10\dB.)}
           \label{fNTxEffect}
   \end{figure}
}

We show in \Fig{fSumRate4Qbit} the sum rate result for partial feedback with quantized \CQI{}.
For the quantized \CQI{} case, we consider $L=1,3,7$ and $15$ in \Fig{fQReg}, each of which corresponds to $1,2,3$ and $4$ bits in quantization $(\NQb \triangleq \lceil \log_2 (L+1) \rceil)$.
We show both the analytical and simulation results for the quantized \CQI{} case.
We find that both results are well matched.
As we can expect, the sum rate increases as the number of bits for quantization increases.
Since we focus on the analytic derivation of the sum rate for partial feedback, we do not optimize the quantization region but use the uniform quantization region, \ie{} $\Fw(\xi_\ell) = \frac{\ell}{L+1}$ for $\Fw(x)$ of \TAS{} and \OSTBC{} in \Section{sRsumQ4Tas} and \Section{sRsumQ4Ostbc}.
Finding the optimal region to maximize the sum rate considering system parameters including diversity type, the number of antennas and users, and the feedback ratio is left as future work.\OneColumn
{
   \begin{figure}
       \centering
           \includegraphics[width=12.0cm]{\PfSchFigDir/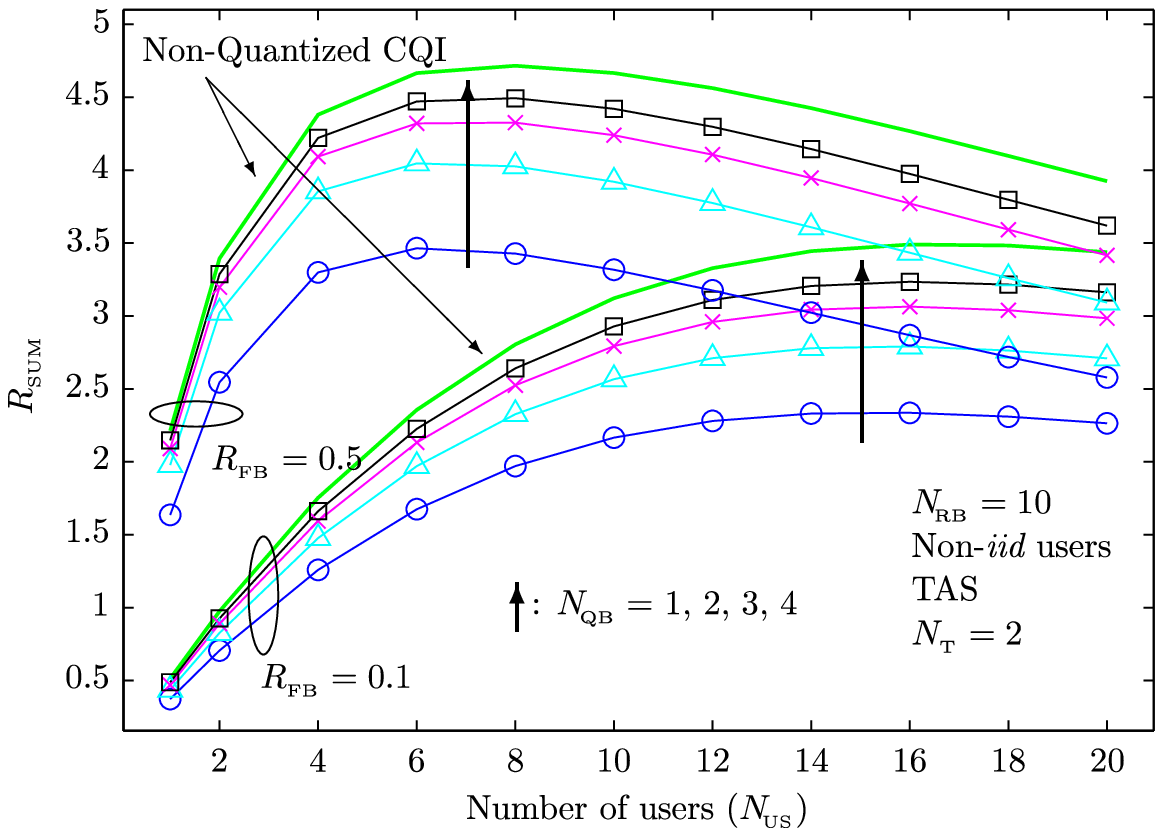}
           \caption{Comparison of the sum rate for non-quantized \CQI{} and quantized \CQI{} for the different feedback ratio. (\TAS{}, $\Nrb=10$, $\lambda$ in \eqref{eUserDist}= 0.0, and Tx \SNR{}= 10\dB.)}
           \label{fSumRate4Qbit}
   \end{figure}
}\TwoColumn
{
   \begin{figure}
       \centering
           \includegraphics[width=8.0cm]{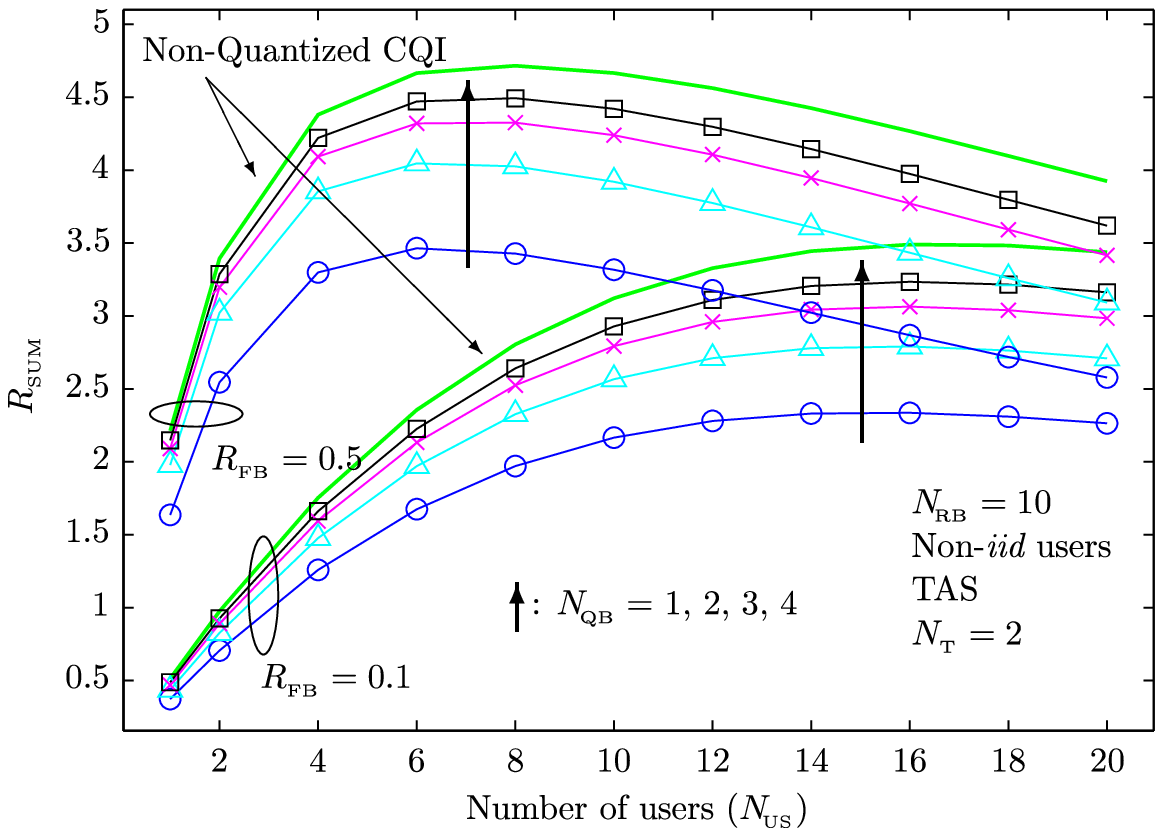}
           \caption{Comparison of the sum rate for non-quantized \CQI{} and quantized \CQI{} for the different feedback ratio. (\TAS{}, $\Nrb=10$, $\lambda$ in \eqref{eUserDist}= 0.0, and Tx \SNR{}= 10\dB.)}
           \label{fSumRate4Qbit}
   \end{figure}
}

In \Fig{fFbLoad}, we show the sum rate for quantized \CQI{} with varying feedback loads.
The feedback load is defined as the number of bits to be sent back from each user, \ie{} $\LFb = \NFb (\lceil \log_2 \Nrb \rceil + \lceil \log_2 \NTx \rceil + \NQb)$.
In the figure, we compare two cases for every fixed $\LFb$ at $12,24$ and $64$ where one of $\NFb$, $\NTx$ or $\NQb$ is additionally fixed.
Specifically, when $\NFb$ is fixed at $8$ in case of $\LFb=64$, we note that the larger $\NTx$ is always preferable to the larger $\NQb$.
When $\NFb$ is made variable, for both $\LFb=12$ and $\LFb=24$ we note that the larger $\NFb$ is preferable for the small population and the larger $\NTx$ or $\NQb$ is preferable for the large population.
This suggests that $\NFb$ should be first determined based on the number of users as in \eqref{eReqRfb} and then based on the value for $\NTx$ the number of feedback bits $\NQb$ should be determined.\OneColumn
{
   \begin{figure}
       \centering
           \includegraphics[width=12.0cm]{\PfSchFigDir/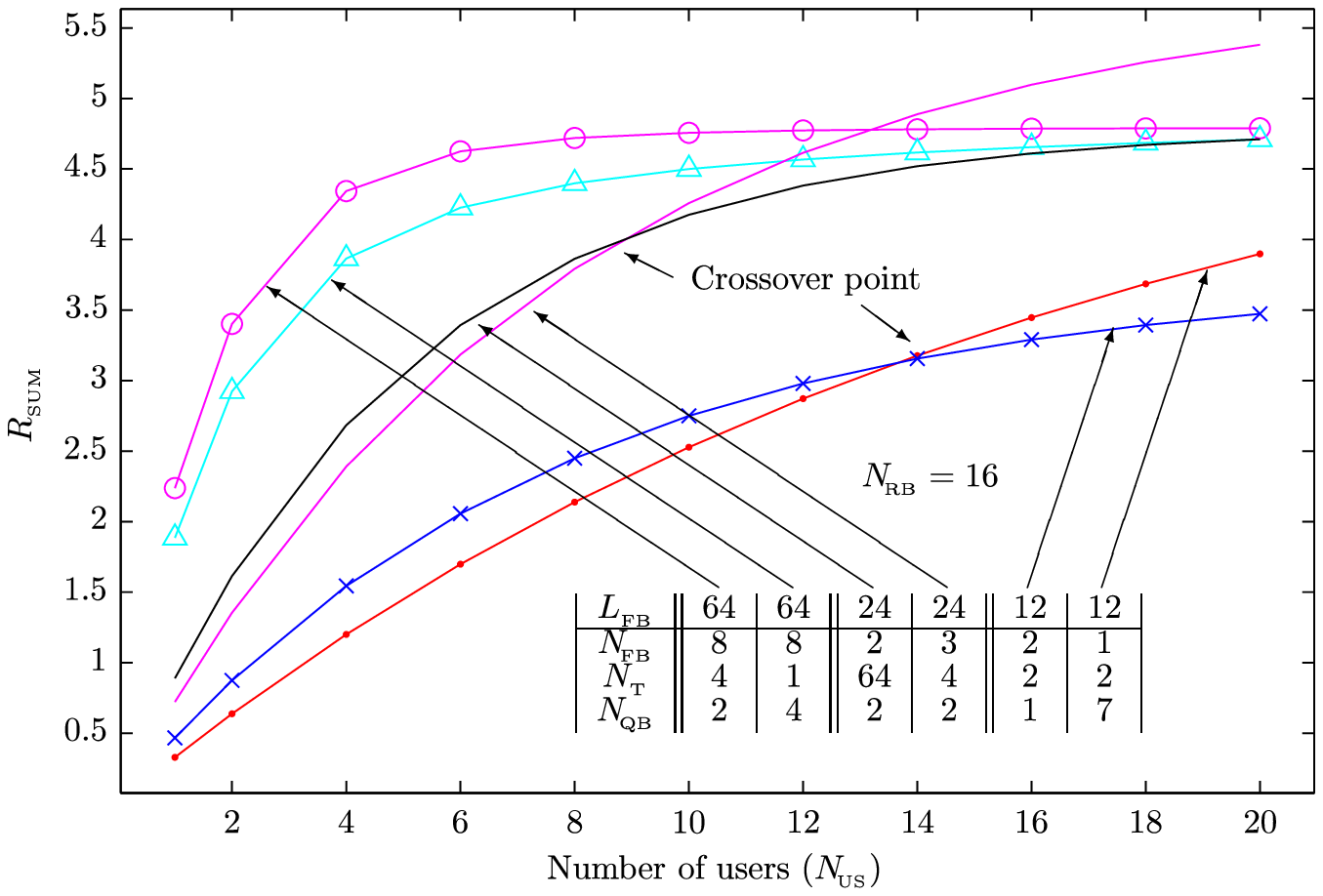}
           \caption{Comparison of the sum rate for the fixed feedback load where $\LFb = \NFb (4 + \lceil \log_2 \NTx \rceil + \NQb)$. (\TAS{}, $\Nrb=16$, $\lambda$ in \eqref{eUserDist}= 0.0, and Tx \SNR{}= 10\dB.)}
           \label{fFbLoad}
   \end{figure}
}\TwoColumn
{
   \begin{figure}
       \centering
           \includegraphics[width=8.0cm]{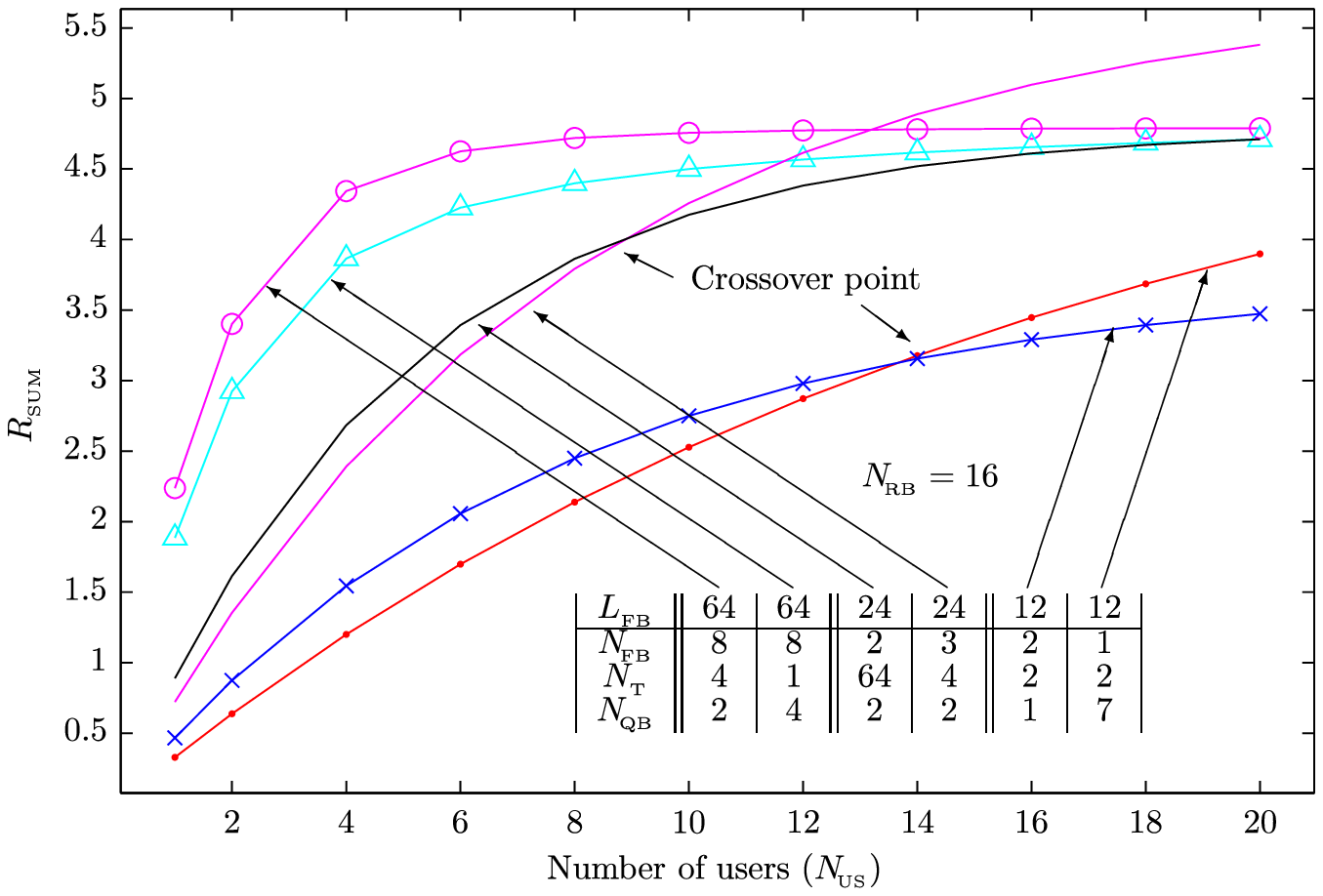}
           \caption{Comparison of the sum rate for the fixed feedback load where $\LFb = \NFb (4 + \lceil \log_2 \NTx \rceil + \NQb)$. (\TAS{}, $\Nrb=16$, $\lambda$ in \eqref{eUserDist}= 0.0, and Tx \SNR{}= 10\dB.)}
           \label{fFbLoad}
   \end{figure}
}

\subsection{The sum rate ratio and required feedback ratio}
In \Fig{fNormRsum}, we study the $\RsumRatio$, \ie{} the sum rate normalized by that of a full feedback scheme as a function of the feedback ratio.
As we expect, the feedback ratio required to achieve a large sum rate ratio decreases with increasing number of users.
We note that the sum rate ratio does not depend on the transmit antenna scheme (\ie{} \TAS{} or \OSTBC{}) and user distribution (\ie{} $\lambda$).
In \Fig{fNormSumPsch}, we can verify the tight relation between the sum rate ratio and the probability of the complement of a scheduling outage when the number of users is not so small.
These two figures support the approximation for the sum rate ratio in \Section{sRelation}, which states that the sum rate ratio is affected mainly by a scheduling outage which is caused when no user provides \CQI{} for a block and that the probability of a scheduling outage depends only on the number of users and the feedback ratio as in \eqref{eSumRateRatio}.
In \Fig{fNormSumPsch}, we also note that the sum rate ratio in the small population (\ie{} $\Nus=2$) moves toward the approximation when the number of antennas increases since the approximation for $\FnRsumNq (x,y,z)$ holds better for larger $\NTx$ especially for \TAS{}.\OneColumn
{
   \begin{figure}
       \centering
           \includegraphics[width=12.0cm]{\PfSchFigDir/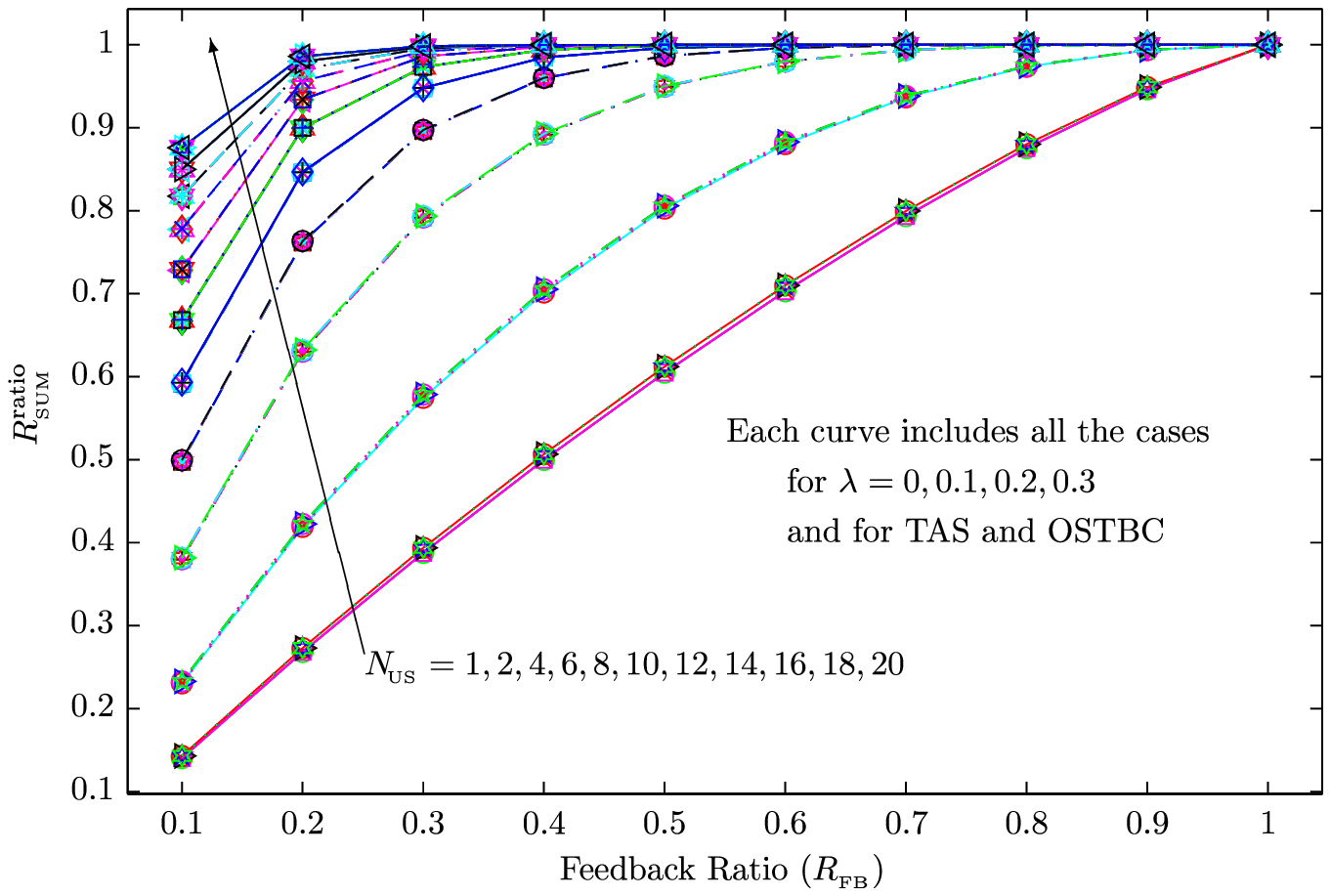}
           \caption{$\Rsum$ normalized by that of a full feedback scheme vs. feedback ratio. We note that the normalized values are independent of transmit antenna scheme (\TAS{} or \OSTBC{}) and user distribution (Slopes).}
           \label{fNormRsum}
   \end{figure}
   \begin{figure}
       \centering
           \includegraphics[width=12.0cm]{\PfSchFigDir/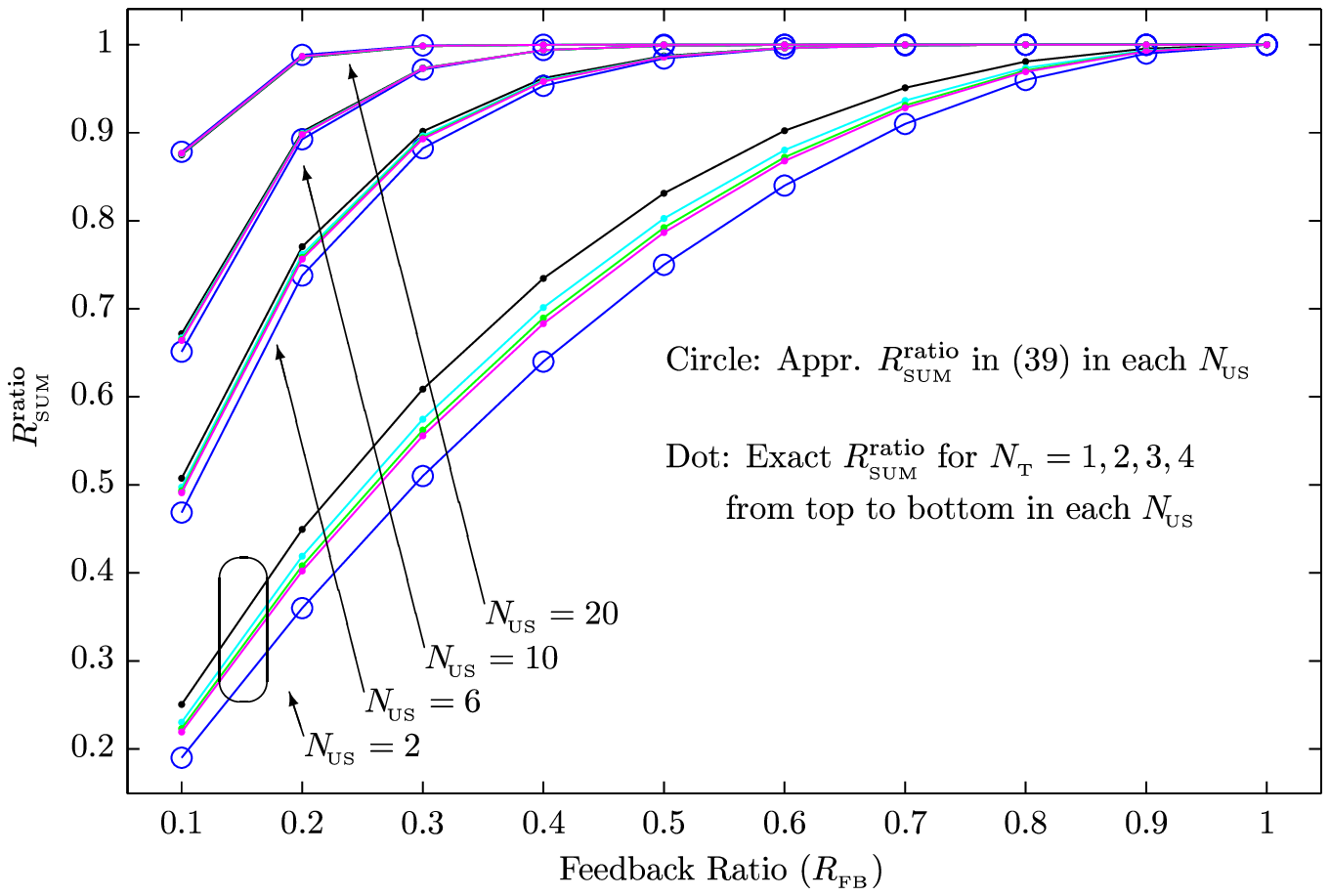}
           \caption{$\Rsum$ normalized by that of a full feedback scheme and the probability of normal scheduling vs. feedback ratio.}
           \label{fNormSumPsch}
   \end{figure}
}\TwoColumn
{
   \begin{figure}
       \centering
           \includegraphics[width=8.0cm]{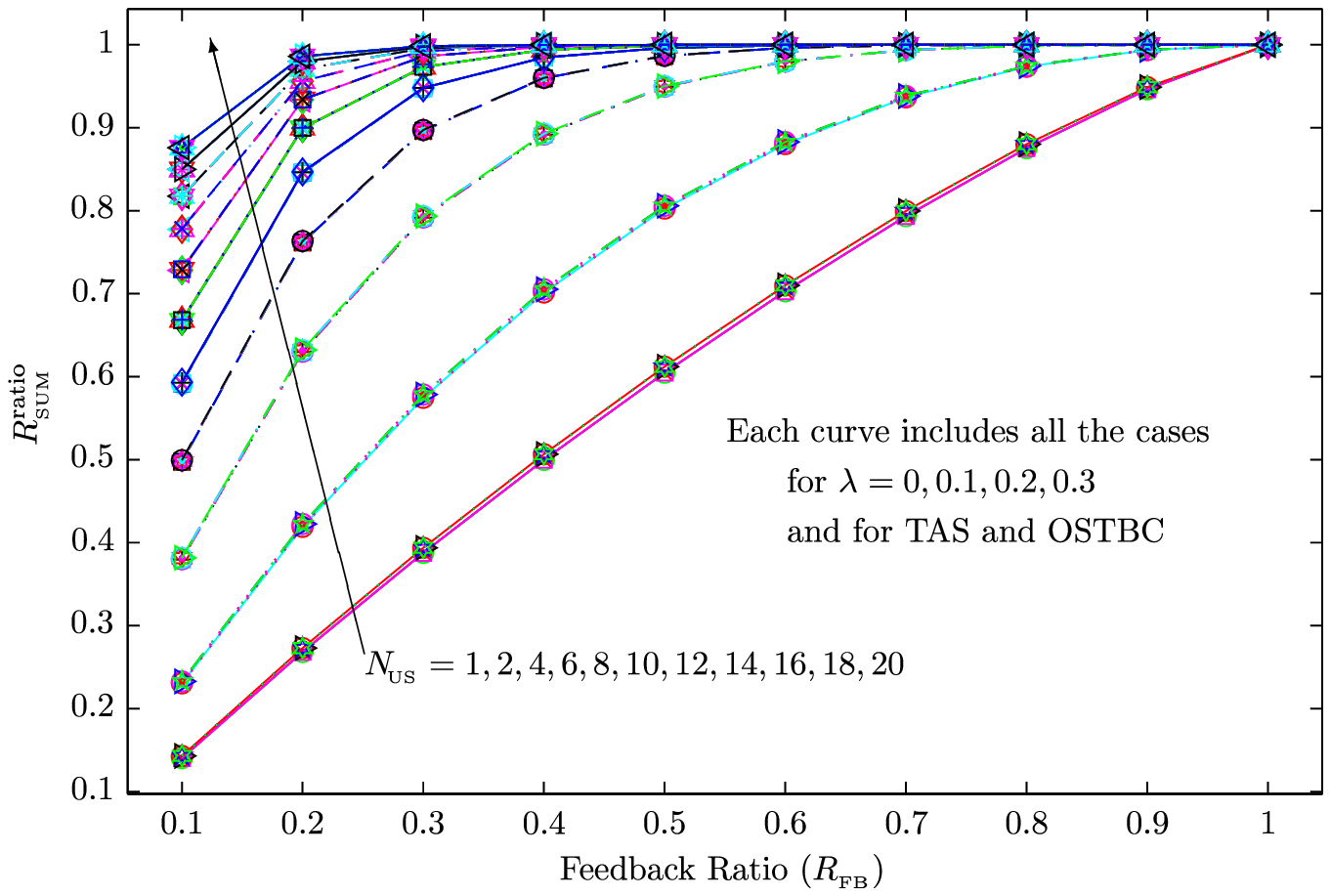}
           \caption{$\Rsum$ normalized by that of a full feedback scheme vs. feedback ratio. We note that the normalized values are independent of transmit antenna scheme (\TAS{} or \OSTBC{}) and user distribution (Slopes).}
           \label{fNormRsum}
   \end{figure}
   \begin{figure}
       \centering
           \includegraphics[width=8.0cm]{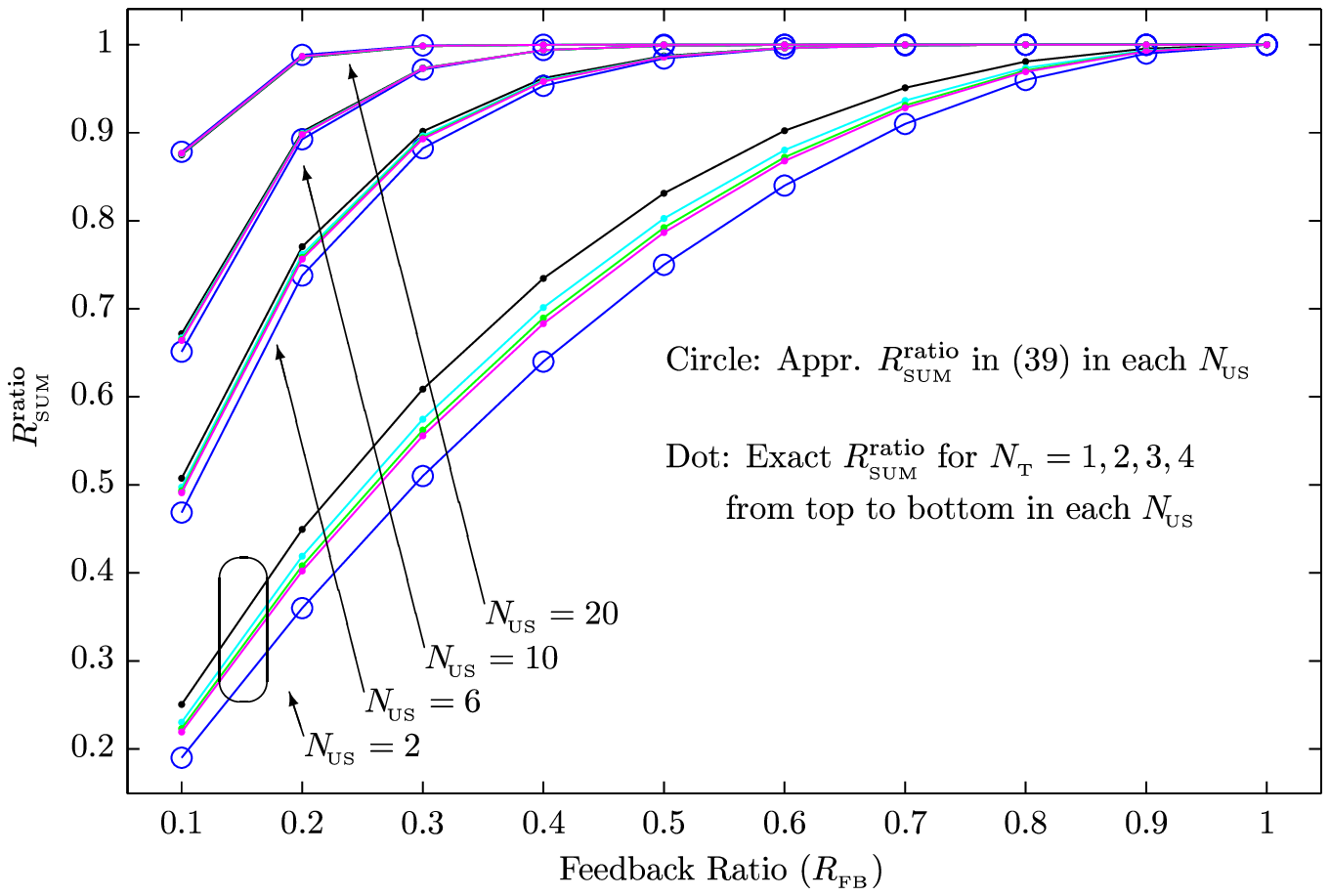}
           \caption{$\Rsum$ normalized by that of a full feedback scheme and the probability of normal scheduling vs. feedback ratio.}
           \label{fNormSumPsch}
   \end{figure}
}

In \Fig{fReqRFb}, we show the required feedback ratio to achieve a pre-determined sum rate ratio.
As the number of users increases, the required feedback ratio decreases because the number of \CQI{} values from all users increases and the scheduling outage probability decreases.
On the other hand, we see that the required feedback ratio increases with the threshold for the smaller scheduling outage probability.
We also note that the required feedback ratio is nearly independent of the transmit antenna scheme and the user distribution.
That is, the required feedback ratio is mainly dependent on the number of users.
Consequently, using this relation, we can determine the appropriate feedback ratio in designing a system.\OneColumn
{
   \begin{figure}
       \centering
           \includegraphics[width=12.0cm]{\PfSchFigDir/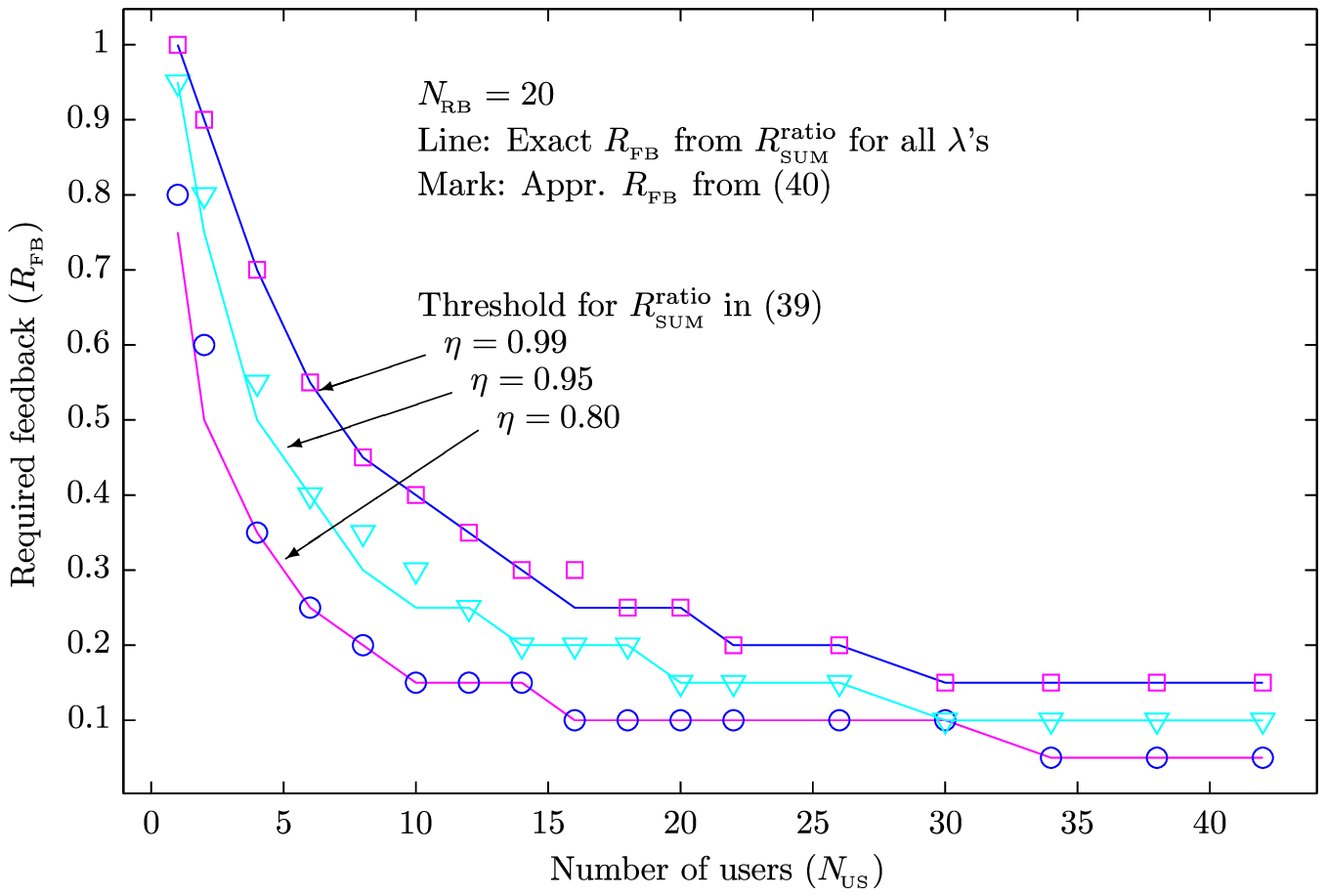}
           \caption{Required feedback ratio to achieve a pre-determined sum rate compared to that by a full feedback scheme.}
           \label{fReqRFb}
   \end{figure}
}\TwoColumn
{
   \begin{figure}
       \centering
           \includegraphics[width=8.0cm]{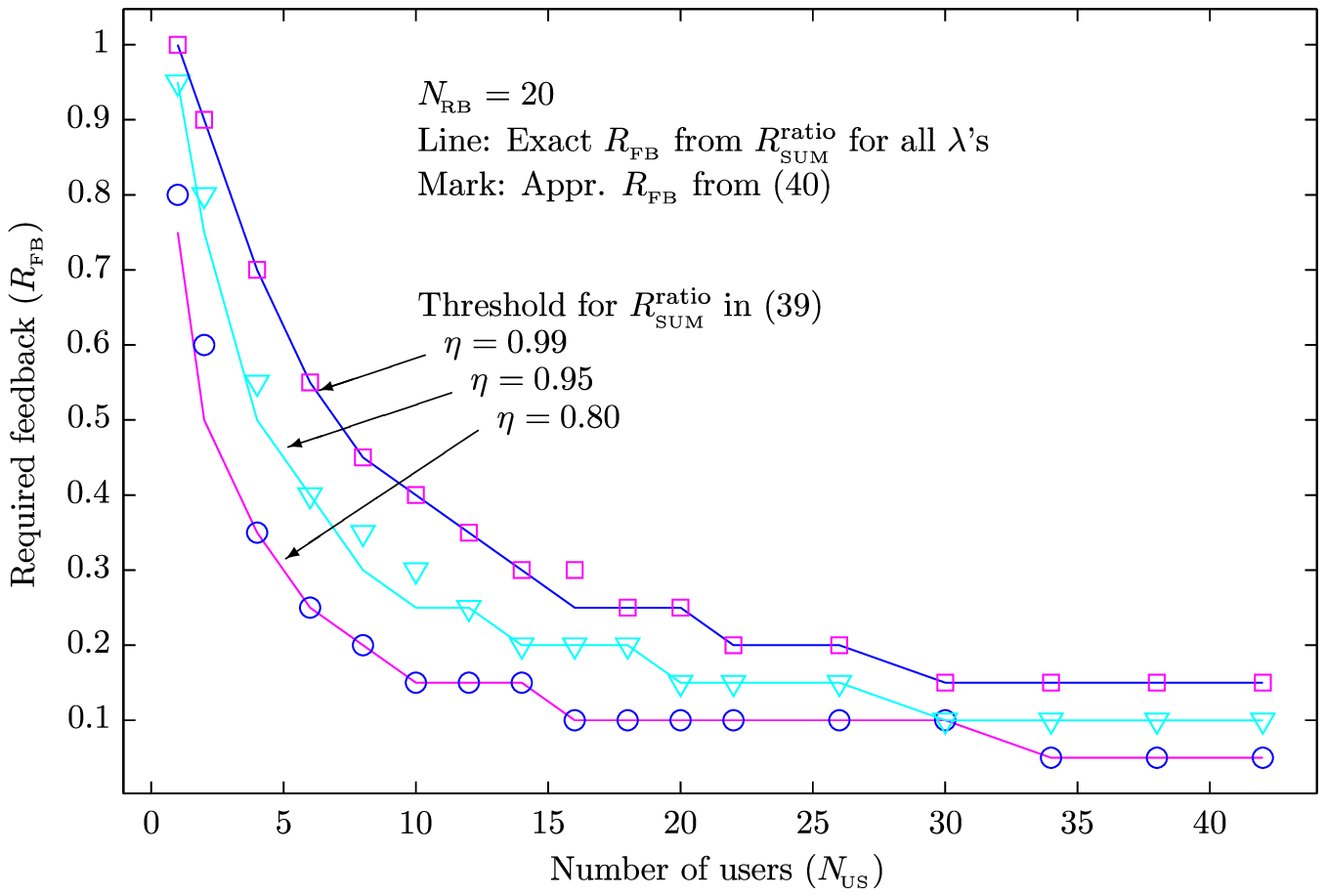}
           \caption{Required feedback ratio to achieve a pre-determined sum rate compared to that by a full feedback scheme.}
           \label{fReqRFb}
   \end{figure}
}

\section{Conclusion}
\label{sConclusion4RsumAnal}
We considered joint scheduling and diversity to enhance the benefits of multiuser diversity in a multiuser \OFDMA{} scheduling system.
We considered the role of partial feedback and developed a unified framework to analyze the sum rate of reduced feedback schemes employing three different multi-antenna transmitter schemes; Transmit antenna selection (\TAS), orthogonal space time block codes (\OSTBC) and cyclic delay diversity (\CDD).
Specifically, for the reduced feedback scheme wherein each user feeds back the best-$\NFb$ \CQI{} values out of a total of $\Nrb$ \CQI{} values, both quantized and non-quantized \CQI{} feedback were addressed.
Considering largest normalized \CQI{} scheduling in each block, closed-form expressions were derived for the sum rate for all the three multi-antenna transmitter schemes.
Further, by approximating the sum rate expression, we derived a simple expression for the minimum required feedback ratio $(\frac{\NFb}{\Nrb})$ to achieve a sum rate comparable to the sum rate obtained by a full feedback scheme.

\appendices
\section{Proof of \Lemma{lFyk}}
\label{sPf4LemFyk}
    The $\Zkr$'s are \iid{} in $r$ and thus $\Ykr$'s in \eqref{eYk} are \iid{} in $r$, which leads to the simplification in notation $\Fyk(x) \triangleq \Fykr(x) = \Pr \{ \Ykr \leq x \}$ and $\Fzk(x) \triangleq {{F_{{\hspace{-0.1cm} {}_{X_{k,r}}}}}}$.
    For additional simplicity in derivation, we first consider the case $\rho=1$ in \eqref{eYk}.
    Since $\Ykr$ is selected among best-$\MinRankingInPf$ random variables, using Bayes' rule \cite{bGarciaAw94}, we have\OneColumn
{
        \begin{equation*}
            \label{eFy04}
            \Fyk(x) = \sum_{m=\NcolInPf-\MinRankingInPf+1}^\NcolInPf \Pr \{ \Ykr = Z_{k,(m)} \} \Pr \{ \Ykr \leq x | \Ykr = Z_{k,(m)}\}.
        \end{equation*}
}\TwoColumn
{
        \begin{equation*}
            \hspace{-1.0cm} \Fyk(x) = \sum_{m=\NcolInPf-\MinRankingInPf+1}^\NcolInPf \Pr \{ \Ykr = Z_{k,(m)} \}
        \end{equation*}
        \begin{equation*}
            \label{eFy04}
            \hspace{3.0cm} \times \Pr \{ \Ykr \leq x | \Ykr = Z_{k,(m)}\}.
        \end{equation*}
}
    We note that $\Pr \{ \Ykr \leq x | \Ykr = Z_{k,(m)}\} = \Pr\{ Z_{k,(m)} \leq x \} = I_{\Fzk(x)} (m,\NcolInPf-m+1)$, where $I_x(\cdot,\cdot)$ denotes an incomplete Beta function \cite[2.1.5]{bDavidJwas04}, and that $\Pr \{ \Ykr = Z_{k,(m)} \} = \Pr \{ \Rkr = m \} = \frac{1}{\MinRankingInPf}$.
    With a suitable change of variables followed by using a summation form of the incomplete Beta function \cite[2.1.3]{bDavidJwas04}, we have\OneColumn
{
        \begin{equation}
            \label{eFyInSumBetaFnc}
            \hspace{-0.0cm} \Fyk(x)= \frac{1}{\MinRankingInPf} \sum_{m=1}^{\MinRankingInPf} \sum_{\ell=\NcolInPf - m + 1}^{\NcolInPf} \binom{\NcolInPf}{\ell} \{ \Fzk (x) \}^{\ell} \{ 1 - \Fzk(x) \}^{\NcolInPf - \ell}.
        \end{equation}
}\TwoColumn
{
        \begin{equation}
            \label{eFyInSumBetaFnc}
            \hspace{-0.0cm} \Fyk(x)= \sum_{m=1}^{\MinRankingInPf} \sum_{\ell=\NcolInPf - m + 1}^{\NcolInPf} \frac{\binom{\NcolInPf}{\ell}}{\MinRankingInPf} \{ \Fzk (x) \}^{\ell} \{ 1 - \Fzk(x) \}^{\NcolInPf - \ell}.
        \end{equation}
}
    We note in \eqref{eFyInSumBetaFnc} that $\Fyk(x)$ is a polynomial form of $\Fzk(x)$.
    Finding a coefficient for each power of $\Fzk(x)$, we can more directly represent $\Fyk(x)$ in terms of a polynomial in $\Fzk(x)$, a form suitable for the subsequent analysis.
    Thus, our purpose is to find the coefficients for those terms.
    Then we have
        \begin{equation}
            \label{eFy01}
            \hspace{-0.0cm} \Fyk(x) \stackrel{(a)}{=} \sum_{\ell=0}^{\MinRankingInPf-1} \frac{\MinRankingInPf - \ell}{\MinRankingInPf}  \binom{\NcolInPf}{\ell} \{ \Fzk (x) \}^{\NcolInPf - \ell} \{ 1 - \Fzk (x) \}^{\ell}
        \end{equation}
        \begin{equation}
            \label{eFy02}
            \hspace{0.1cm} \stackrel{(b)}{=} \sum_{\ell=0}^{\MinRankingInPf-1} \sum_{r=0}^{\ell} \frac{\MinRankingInPf - \ell}{\MinRankingInPf} \binom{\NcolInPf}{\ell} \binom{\ell}{r} (-1)^r \{ \Fzk (x) \}^{\NcolInPf - \ell + r}
        \end{equation}
        \begin{equation}
            \label{eFy03}
            \hspace{-0.2cm} \stackrel{(c)}{=} \sum_{m=0}^{\MinRankingInPf-1} \sum_{\ell=m}^{\MinRankingInPf-1} \frac{\MinRankingInPf - \ell}{\MinRankingInPf} \binom{\NcolInPf}{\ell} \binom{\ell}{m} (-1)^{\ell - m} \{ \Fzk (x) \}^{\NcolInPf - m},
        \end{equation}
    where $(a)$ follows from switching the order of $m$ and $\ell$ in \eqref{eFyInSumBetaFnc} and adjusting $\ell$;
    $(b)$ follows from applying the binomial theorem \cite{bGarciaAw94} to $\{ 1 - \Fzk (x) \}^{\ell}$ in \eqref{eFy01};
    $(c)$ follows from replacing $\ell-r$ with $m$ in \eqref{eFy02} and switching $m$ and $\ell$.
    Since the power of $\Fzk(x)$ is independent of $\ell$ in \eqref{eFy03}, we can represent \eqref{eFy03} as \eqref{eFy00} with $e_1(\NcolInPf,\MinRankingInPf,m)$ given by \eqref{eE1abm} after considering a constant $\rho$.

\section{Proof of \Corollary{cE1}}
\label{sPf4CorrE1}
    When $\MinRankingInPf=\NcolInPf$, \eqref{eE1abm} reduces to $e_1(\NcolInPf,\NcolInPf,m)= \sum_{\ell=m}^{\NcolInPf-1} \binom{\NcolInPf-1}{\ell} \binom{\ell}{m} (-1)^{\ell-m}$.
    When we take the derivative $m$ times with respective to $x$ of the binomial expansion of $(1-x)^{\NcolInPf-1} = \sum_{\ell=0}^{\NcolInPf-1} \binom{\NcolInPf-1}{m} (-1)^\ell x^\ell$ and divide both sides by $m!$, we have
        \begin{equation}
            \label{eBinomExpansion}
            (-1)^m \tbinom{\NcolInPf-1}{m} (1-x)^{\NcolInPf-m-1} = \sum_{\ell=m}^{\NcolInPf-1} \tbinom{\NcolInPf-1}{\ell} \tbinom{\ell}{m} (-1)^\ell x^{\ell - m}.
        \end{equation}
    When we plug $x=1$ in both sides and divide both sides by $(-1)^m$, we can find that $e_1(\NcolInPf,\NcolInPf,m) = \sum_{\ell=m}^{\NcolInPf-1} \binom{\NcolInPf-1}{\ell} \binom{\ell}{m} (-1)^{\ell-m} = 1$ for $m=\NcolInPf-1$, and $0$ otherwise.

\section{Derivation of the conditional \CDF{} of $X_r$}
\label{sDeriCdfX}
Following the notations in \Section{sStepOne4Nq}, since a selected user is $k$ and the number of users who provided \CQI{} to the transmitter is $n$, we have the conditional \CDF{} of $X_r$ as\OneColumn
{
    \begin{equation}
        \label{eFyDeri}
        \hspace{-0.0 cm} \FxCondSr (x) \stackrel{(a)}{=} \Pr \{ X_r \leq x \; | \; k_r^* = k, |S_r|=n \} \stackrel{(b)}{=} \Pr \{ \Ykr \leq x \; | \; k_r^* = k, |S_r|=n \}
    \end{equation}
    \begin{equation*}
        \hspace{-0.0 cm} \stackrel{(c)}{=} \Pr \left \{ \Ukr \leq \frac{x}{\rho c_{k}} \; | \; k_r^* = k, |S_r|=n \right \} \stackrel{(d)}{=} \Pr \left \{ U_{i,r} \leq \frac{x}{\rho c_{k}} \; , \; \forall i \in S_r \; | \; |S_r|=n \right \}
    \end{equation*}
    \begin{equation*}
        \hspace{-0.0 cm} \stackrel{(e)}{=} \prod_{i \in S_r, |S_r|=n } \Pr \left \{ U_{i,r} \leq \frac{x}{\rho c_{k}} \right \} \stackrel{(f)}{=} \prod_{i \in S_r, |S_r|=n} \Pr \{ \rho c_{k} \Ukr \leq x \} \stackrel{(g)}{=} \prod_{i \in S_r, |S_r|=n} \Fyk (x) = \left \{ \Fyk (x) \right \}^{n},
    \end{equation*}
}\TwoColumn
{
    \begin{equation}
        \label{eFyDeri}
        \hspace{-0.0 cm} \FxCondSr (x) \stackrel{(a)}{=} \Pr \{ X_r \leq x \; | \; k_r^* = k, |S_r|=n \}
    \end{equation}
    \begin{equation*}
        \hspace{-4.0 cm} \stackrel{(b)}{=} \Pr \{ \Ykr \leq x \; | \; k_r^* = k, |S_r|=n \}
    \end{equation*}
    \begin{equation*}
        \hspace{-3.3 cm} \stackrel{(c)}{=} \Pr \left \{ \Ukr \leq \frac{x}{\rho c_{k}} \; | \; k_r^* = k, |S_r|=n \right \}
    \end{equation*}
    \begin{equation*}
        \hspace{-3.2 cm} \stackrel{(d)}{=} \Pr \left \{ U_{i,r} \leq \frac{x}{\rho c_{k}} \; , \; \forall i \in S_r \; | \; |S_r|=n \right \}
    \end{equation*}
    \begin{equation*}
        \hspace{-0.2 cm} \stackrel{(e)}{=} \prod_{i \in S_r, |S_r|=n } \Pr \left \{ U_{i,r} \leq \frac{x}{\rho c_{k}} \right \} \stackrel{(f)}{=} \prod_{i \in S_r, |S_r|=n} \Pr \{ \rho c_{k} \Ukr \leq x \}
    \end{equation*}
    \begin{equation*}
        \hspace{-4.0 cm} \stackrel{(g)}{=} \prod_{i \in S_r, |S_r|=n} \Fyk (x) = \left \{ \Fyk (x) \right \}^{n},
    \end{equation*}
}
where $(a)$ follows from the definition of \CDF; $(b)$ from $X_r = \Ykr$ because user-$k$ is selected; $(c)$ from the definition of $\Ukr$; $(d)$ from that $\Ukr$ is the maximum among users in $S_r$; $(e)$ from \iid{} property of $U_{i,r}$ in $i$; $(f)$ from the identical distribution of $U_{i,r}$ in $i$; $(g)$ from the definition of $\Ykr$ and its \CDF{}.

\section{Proof of \Lemma{lFxk}}
\label{sPf4LemFxk}
    From \eqref{eFy00} and \eqref{eFx00}, we have\OneColumn
{
        \begin{equation}
            \label{eMaxNfbFyk01}
            \hspace{0.0cm} \FxCondSr(x) = \{ \Fyk (x) \}^n = \{ \Fzk (\tfrac{x}{\rho}) \}^{\NrowVal \NcolInPf} \left \{ \sum_{m=0}^{\MinRankingInPf-1} \frac{e_1(\NcolInPf,\MinRankingInPf,m)}{\{ \Fzk (\tfrac{x}{\rho}) \}^{m}} \right \}^\NrowVal.
        \end{equation}
}\TwoColumn
{
        \begin{equation*}
            \hspace{-3.0cm} \FxCondSr(x) = \{ \Fyk (x) \}^n
        \end{equation*}
        \begin{equation}
            \label{eMaxNfbFyk01}
            \hspace{0.0cm} = \{ \Fzk (\tfrac{x}{\rho}) \}^{\NrowVal \NcolInPf} \left \{ \sum_{m=0}^{\MinRankingInPf-1} \frac{e_1(\NcolInPf,\MinRankingInPf,m)}{\{ \Fzk (\tfrac{x}{\rho}) \}^{m}} \right \}^\NrowVal.
        \end{equation}
}
    Applying the same technique as in \cite[0.314]{bGradshteynAp00} and \cite[(16)]{bChenTcom06ForLetter} to a finite-order polynomial, we can express the above equation in a polynomial form and compute the coefficients for each term.
    More specifically, regarding \eqref{eMaxNfbFyk01} as a polynomial in $\frac{1}{\Fzk(\tfrac{x}{\rho})}$, we can calculate coefficients for $\frac{1}{\Fzk(\tfrac{x}{\rho})}$ in a recursive form as given by \eqref{eE2}, and $\FxCondSr(x)$ has the form given by \eqref{eCdfX}.

\section{Derivation of $\FnRsumNq (x,y,z)$}
\label{sDeriC1}
Following the approach in \cite{bChenTcom06}, we can compute $\FnRsumNq (x,y,z)$ in \eqref{eC1Ori}.
We note that the final form we have in \eqref{eC1Ori} is much better than that in \cite[(15), (42)]{bChenTcom06} in evaluating large values for the arguments.

The \PDF{} of $Z$ which follows the Gamma distribution with $\Gcal(\alpha,\beta)$ is given by $\fz(z)= \frac{\beta^\alpha}{\Gamma(\alpha)} z^{\alpha-1} e^{-\beta z}$ from the derivative of \CDF{} in \eqref{eGammaCdf}.
When $\alpha$ is a positive integer, the \CDF{} in \eqref{eGammaCdf} is represented by direct integration as $\Fz(z) = 1 - e^{- \beta z} \sum_{i=0}^{\alpha-1} \frac{(\beta z)^i}{i!}$.
Since $ d \{ \Fz(z)\}^n = n \{ \Fz(z) \}^{n-1} \fz(z) dz$, we have from \cite[(18)]{bFedeleETcom96} and \cite[(40)]{bChenTcom06}\OneColumn
{
    \begin{equation}
        d \{ \Fz(z)\}^n = \tfrac{n}{(\alpha-1)!} \sum_{k=0}^{n-1} (-1)^{k} \tbinom{n-1}{k} \sum_{i=0}^{k(\alpha-1)} b_{k,i} \beta^{\alpha+i} e^{-(k+1) \beta z} z^{\alpha + i -1} dz
    \end{equation}
}\TwoColumn
{
    \begin{equation*}
        \hspace{-3.0cm} d \{ \Fz(z)\}^n = \tfrac{n}{(\alpha-1)!} \sum_{k=0}^{n-1} (-1)^{k} \tbinom{n-1}{k}
    \end{equation*}
    \begin{equation}
        \hspace{1.0cm} \times \sum_{i=0}^{k(\alpha-1)} b_{k,i} \beta^{\alpha+i} e^{-(k+1) \beta z} z^{\alpha + i -1}
    \end{equation}
}
for $b_{k,i}$ in \eqref{ebki}.
Then, using the integration identity $\int_0^\infty z^{n-1} e^{-x z} \ln(1+z) dz = (n-1)! e^x \sum_{\ell=1}^{n} \frac{\Gamma(\ell-n,x)}{x^\ell}$ \cite[(78)]{bAlouiniTvt99}, we have for $\FnRsumNq (\alpha,\beta,n)= \int_0^\infty \log(1+z) d \{ \Fz(z) \}^n$ as \cite[(42)]{bChenTcom06}\OneColumn
{
    \begin{equation}
        \label{eC1Ref}
        \tfrac{n}{(\alpha-1)! \ln 2} \sum_{k=0}^{n-1} (-1)^{k} \tbinom{n-1}{k} \sum_{i=0}^{k(\alpha-1)} b_{k,i} \beta^{\alpha+i} e^{(k+1) \beta} (\alpha+i-1)! \sum_{\ell=1}^{\alpha+i} \left[ \tfrac{1}{(k+1) \beta} \right]^\ell \Gamma(\ell-\alpha - i, (k+1) \beta).
    \end{equation}
}\TwoColumn
{
    \begin{equation*}
        \tfrac{n}{(\alpha-1)! \ln 2} \sum_{k=0}^{n-1} (-1)^{k} \tbinom{n-1}{k} \sum_{i=0}^{k(\alpha-1)} b_{k,i} \beta^{\alpha+i} e^{(k+1) \beta} (\alpha+i-1)!
    \end{equation*}
    \begin{equation}
        \label{eC1Ref}
        \times \sum_{\ell=1}^{\alpha+i} \left[ \tfrac{1}{(k+1) \beta} \right]^\ell \Gamma(\ell-\alpha - i, (k+1) \beta).
    \end{equation}
}
By adjusting summation index for $\ell$ and replacing $\alpha$, $\beta$, and $n$ with $x$, $y$, and $z$ respectively, we can have \eqref{eC1Ori}.
When $\alpha=1$, we follow the same procedure and use the integration identity $\int_0^\infty e^{-xt} \ln(1+yt) dt \stackrel{(a)}{=} \frac{1}{x} e^{\frac{x}{y}} \int_{\frac{x}{y}}^\infty \frac{e^t}{t} dt \stackrel{(b)}{=} \Gamma(0,\frac{x}{y})$ to obtain \eqref{eC1Red}, where $(a)$ follows from \cite[4.337.2, 8.211.1]{bGradshteynAp00} and $(b)$ follows from \cite[8.350.2]{bGradshteynAp00}.

\section{Proof of \iid{} property for $\UkrQuan$}
\label{sStatOfUkrQuan}
Since $\UkrQuan$ is equivalent to a quantized value of $\Ukr$ by the policy in \eqref{eQuanPolicy}, we have\OneColumn
{
    \begin{equation}
        \label{eIdenticalDistriOfUkrQuan}
        \Pr \{ \UkrQuan = J_\ell \} \stackrel{(a)}{=} \Pr \{ \xi_{\ell} \leq \Ukr < \xi_{\ell+1} \} \stackrel{(b)}{=} \Pr \{ \xi_{\ell} \leq \Umn < \xi_{\ell+1} \} \stackrel{(c)}{=} \Pr \{ \UmnQuan = J_\ell \}
    \end{equation}
}\TwoColumn
{
    \begin{equation*}
        \hspace{-3.0cm} \Pr \{ \UkrQuan = J_\ell \} \stackrel{(a)}{=} \Pr \{ \xi_{\ell} \leq \Ukr < \xi_{\ell+1} \}
    \end{equation*}
    \begin{equation}
        \label{eIdenticalDistriOfUkrQuan}
        \stackrel{(b)}{=} \Pr \{ \xi_{\ell} \leq \Umn < \xi_{\ell+1} \} \stackrel{(c)}{=} \Pr \{ \UmnQuan = J_\ell \}
    \end{equation}
}
where $(a)$ and $(c)$ follows from the quantization policy in \eqref{eQuanPolicy} and $(b)$ follows that $\Ukr$ is identically distributed. Therefore, $\UkrQuan$ is identically distributed.
Further, we have\OneColumn
{
    \begin{equation*}
        \hspace{-5.0cm} \Pr \Big \{ \bigcap_{k=1}^{\Nus} \bigcap_{r=1}^{\Nrb} \UkrQuan = J_{\ell_{k,r}} \Big \} \stackrel{(a)}{=} \Pr \Big \{ \bigcap_{k=1}^{\Nus} \bigcap_{r=1}^{\Nrb} \xi_{\ell_{k,r}} \leq \Ukr < \xi_{\ell_{k,r}+1} \Big \}
    \end{equation*}
    \begin{equation}
        \label{eIndDistriOfUkrQuan}
        \hspace{3.5cm} \stackrel{(b)}{=} \prod_{k=1}^{\Nus} \prod_{r=1}^{\Nrb} \Pr \{ \xi_{\ell_{k,r}} \leq \Ukr < \xi_{\ell_{k,r}+1} \} \stackrel{(c)}{=} \prod_{k=1}^{\Nus} \prod_{r=1}^{\Nrb} \Pr \{ \UkrQuan = J_{\ell_{k,r}} \}
    \end{equation}
}\TwoColumn
{
    \begin{equation*}
        \hspace{-4.0cm} \Pr \Big \{ \bigcap_{k=1}^{\Nus} \bigcap_{r=1}^{\Nrb} \UkrQuan = J_{\ell_{k,r}} \Big \}
    \end{equation*}
    \begin{equation*}
        \hspace{-2.0cm} \stackrel{(a)}{=} \Pr \Big \{ \bigcap_{k=1}^{\Nus} \bigcap_{r=1}^{\Nrb} \xi_{\ell_{k,r}} \leq \Ukr < \xi_{\ell_{k,r}+1} \Big \}
    \end{equation*}
    \begin{equation*}
        \hspace{-2.3cm} \stackrel{(b)}{=} \prod_{k=1}^{\Nus} \prod_{r=1}^{\Nrb} \Pr \{ \xi_{\ell_{k,r}} \leq \Ukr < \xi_{\ell_{k,r}+1} \}
    \end{equation*}
    \begin{equation}
        \label{eIndDistriOfUkrQuan}
        \hspace{-3.8cm} \stackrel{(c)}{=} \prod_{k=1}^{\Nus} \prod_{r=1}^{\Nrb} \Pr \{ \UkrQuan = J_{\ell_{k,r}} \}
    \end{equation}
}
where $(a)$ and $(c)$ follows from the quantization policy in \eqref{eQuanPolicy} and $(b)$ follows that $\Ukr$ is independent. Therefore, $\UkrQuan$ is independent. From \eqref{eIdenticalDistriOfUkrQuan} and \eqref{eIndDistriOfUkrQuan}, we find that $\UkrQuan$ is \iid.

\section{Derivation of the conditional \PMF}
\label{sDerivePmf}
Following the notations in \Section{sRsumQ4Tas}, let us suppose that $n$ users provided the quantization index at block-$r$.
The probability that the quantization index of a selected user is $J_\ell$ is the same as the probability that the maximum of $\UkrQuan$ for all users is $J_\ell$.
Thus it is given by\OneColumn
{
    \begin{equation}
        \label{eQuanIndProb}
        \hspace{-0.0cm} \Pr\{ J_\ell \textrm{ is selected} \hspace{0.1cm} | \hspace{0.1cm} |S_r|= n \} = \Pr \big\{ \max_{k' \in S_r} U_{k',r}^\QuanSup \leq J_\ell \big \} - \Pr \big\{ \max_{k' \in S_r} U_{k',r}^\QuanSup \leq J_{\ell-1} \big \}
    \end{equation}
    \begin{equation*}
        \hspace{0.0cm} \stackrel{(a)}{=} \Pr \big\{ \max_{k' \in S_r} U_{k',r} \leq \xi_{\ell + 1} \big \} - \Pr \big\{ \max_{k' \in S_r} U_{k',r} \leq \xi_{\ell} \big \} \stackrel{(b)}{=} \left\{ \Fu(\xi_{\ell+1}) \right\}^n -  \left\{\Fu(\xi_\ell) \right\}^n
    \end{equation*}
}\TwoColumn
{
    \begin{equation}
        \label{eQuanIndProb}
        \hspace{-4.0cm} \Pr\{ J_\ell \textrm{ is selected} \hspace{0.1cm} | \hspace{0.1cm} |S_r|= n \}
    \end{equation}
    \begin{equation*}
        \hspace{-0.0cm} = \Pr \big\{ \max_{k' \in S_r} U_{k',r}^\QuanSup \leq J_\ell \big \} - \Pr \big\{ \max_{k' \in S_r} U_{k',r}^\QuanSup \leq J_{\ell-1} \big \}
    \end{equation*}
    \begin{equation*}
        \hspace{0.0cm} \stackrel{(a)}{=} \Pr \big\{ \max_{k' \in S_r} U_{k',r} \leq \xi_{\ell + 1} \big \} - \Pr \big\{ \max_{k' \in S_r} U_{k',r} \leq \xi_{\ell} \big \}
    \end{equation*}
    \begin{equation*}
        \hspace{-3.5cm} \stackrel{(b)}{=} \left\{ \Fu(\xi_{\ell+1}) \right\}^n -  \left\{\Fu(\xi_\ell) \right\}^n
    \end{equation*}
}
where $(a)$ follows from the quantization policy in \eqref{eQuanPolicy} and $(b)$ follows from the order statistics \cite[2.1.1]{bDavidJwas04}.
Since user selection is based on \iid{} normalized \CQI{} values, the probability that each user is selected for a transmission is $\frac{1}{\Nus}$.
Considering that the modulation level is determined as $\rho c_k \xi_\ell$ for user-$k$ when it is selected, the conditional \PMF{} that $\XrQuan = \rho c_k \xi_\ell$ is given by \eqref{eCondPmf}.
We note that the sum of this probability over $n$ and $\ell$ is 1, which verifies the validity as the \PMF.

\section{Derivation of \eqref{eRsumQ}}
\label{sDeriveRsumQ}
From the conditional \PMF{} in \eqref{eCondPmf}, the sum rate for the system with partial feedback of quantized \CQI{} is given by\OneColumn
{
    \begin{equation*}
        \hspace{-0.8cm} \Rsum = \frac{1}{\Nrb} \sum_{r=1}^{\Nrb} \Ex [ \log (1 + \XrQuan ) ] \stackrel{(a)}{=} \Ex[\log(1+\XrQuan)] = \Ex_{{}_{|S_r|}} \Ex_{{}_{\XrQuan}}[\log(1+\XrQuan) \; | \; |S_r| = n \neq 0]
    \end{equation*}
    \begin{equation}
        \label{eRsumQIntermediate}
        \hspace{-0.4cm} \stackrel{(b)}{=} \sum_{k=1}^\Nus \sum_{\ell=0}^L \tfrac{\log_2 (1 + \rho c_k \xi_\ell)}{\Nus} \sum_{n=1}^\Nus \tbinom{\Nus}{n} \left( \tfrac{\NFb}{\Nrb} \right)^n \left( 1 - \tfrac{\NFb}{\Nrb} \right)^{\Nus-n} \left[ \left\{ \Fu(\xi_{\ell+1}) \right\}^n -  \left\{\Fu(\xi_\ell) \right\}^n \right],
    \end{equation}
}\TwoColumn
{
    \begin{equation*}
        \hspace{-0.8cm} \Rsum = \frac{1}{\Nrb} \sum_{r=1}^{\Nrb} \Ex [ \log (1 + \XrQuan ) ] \stackrel{(a)}{=} \Ex[\log(1+\XrQuan)]
    \end{equation*}
    \begin{equation*}
        \hspace{-0.8cm} = \Ex_{{}_{|S_r|}} \Ex_{{}_{\XrQuan}}[\log(1+\XrQuan) \; | \; |S_r| = n \neq 0]
    \end{equation*}
    \begin{equation*}
        \hspace{-0.8cm} \stackrel{(b)}{=} \sum_{k=1}^\Nus \sum_{\ell=0}^L \tfrac{\log_2 (1 + \rho c_k \xi_\ell)}{\Nus} \sum_{n=1}^\Nus \tbinom{\Nus}{n} \left( \tfrac{\NFb}{\Nrb} \right)^n
    \end{equation*}
    \begin{equation}
        \label{eRsumQIntermediate}
        \hspace{1.0cm} \times \left( 1 - \tfrac{\NFb}{\Nrb} \right)^{\Nus-n} \left[ \left\{ \Fu(\xi_{\ell+1}) \right\}^n -  \left\{\Fu(\xi_\ell) \right\}^n \right],
    \end{equation}
}
where $(a)$ follows from that $\XrQuan$ is identically distributed in $r$ and $(b)$ follows from the conditional \PMF{} of $\XrQuan$ in \eqref{eCondPmf} and the \PMF{} of $|S_r|$ in \eqref{ePmfS}.
From the binomial theorem  \cite{bGarciaAw94}, we have\OneColumn
{
    \begin{equation}
        \sum_{n=1}^\Nus \tbinom{\Nus}{n} \left( \tfrac{\NFb}{\Nrb} \right)^n \left( 1 - \tfrac{\NFb}{\Nrb} \right)^{\Nus-n} \left\{ \Fu(\xi_{\ell+1}) \right\}^n = \big \{1 - \tfrac{\NFb}{\Nrb} ( 1 - \Fu(\xi_{\ell+1}))\big \}^\Nus - \Big(1- \tfrac{\NFb}{\Nrb} \Big)^\Nus.
    \end{equation}
}\TwoColumn
{
    \begin{equation*}
        \hspace{-1.0cm} \sum_{n=1}^\Nus \tbinom{\Nus}{n} \left( \tfrac{\NFb}{\Nrb} \right)^n \left( 1 - \tfrac{\NFb}{\Nrb} \right)^{\Nus-n} \left\{ \Fu(\xi_{\ell+1}) \right\}^n
    \end{equation*}
    \begin{equation}
        = \big \{1 - \tfrac{\NFb}{\Nrb} ( 1 - \Fu(\xi_{\ell+1}))\big \}^\Nus - \Big(1- \tfrac{\NFb}{\Nrb} \Big)^\Nus.
    \end{equation}
}
Thus, \eqref{eRsumQIntermediate} reduces to \eqref{eRsumQ} for $\FnRsumQ (x,y,z,r)$ in \eqref{eRsumQ}.

\bibliographystyle{\TexComDir/IEEEbib}
\bibliography{\TexComDir/MyRefs}

\end{document}